\documentclass[a4paper,aps,pra,notitlepage,superscriptaddress,longbibliography,nofootinbib,twocolumn,10pt,reprint]{revtex4-2}
\usepackage[utf8]{inputenc}
\usepackage{amsmath,amsthm,amssymb,amsfonts}
\usepackage{braket}
\usepackage{graphicx}
\usepackage[breaklinks=true]{hyperref}
\hypersetup{
  colorlinks   = true, 
  urlcolor     = blue, 
  linkcolor    = blue, 
  citecolor    = red 
}
\usepackage[bb=boondox]{mathalfa}
\usepackage{yhmath} 
\usepackage{multirow}
\usepackage{mathtools}
\usepackage{soul} 

\usepackage{array}
\newcolumntype{C}[1]{>{\centering\let\newline\\\arraybackslash\hspace{0pt}}m{#1}}

\usepackage{tikz,tikz-3dplot,pgfplots,tikz-cd}
\usetikzlibrary{shapes,arrows,patterns,3d}
\usetikzlibrary{decorations.markings,arrows,decorations.pathreplacing,calc,intersections,through,backgrounds}

\newcommand{\ccaption}[2]{\caption[#1]{\textit{#1.} #2}}

\usepackage{mathdots}

\DeclareMathOperator*{\Moplus}{\text{\raisebox{0.25ex}{\scalebox{0.7}{$\bigoplus$}}}}
\newcommand{\spare}[1]{\left[ {#1} \right]}

\newtheorem*{result1}{Result 1}
\newtheorem*{result2}{Result 2}
\newtheorem*{result3}{Result 3}

\newtheorem{lemma}{Lemma}

\definecolor{gray}{gray}{0.5}
\definecolor{G}{rgb}{0.6,0,0}
\definecolor{J}{rgb}{0.6,0,0.6}
\definecolor{O}{rgb}{0,0,0.6}
\definecolor{dgreen}{rgb}{0,0.6,0}
\definecolor{dred}{rgb}{0.6,0,0}
\definecolor{lred}{rgb}{1,.8,.8}

\newcommand{\ii}{\mathrm{i}}
\newcommand{\ie}{{\it i.e.},\ }
\newcommand{\eg}{{\it e.g.},\ }

\newcommand{\id}{\mathbb{1}} 
\newcommand{\tr}{\operatorname{tr}}
\newcommand{\Tr}{\operatorname{Tr}}

\newcommand{\adrep}[1]{\mathrm{ad}_{#1}}
\newcommand{\ADrep}[1]{\mathrm{Ad}_{#1}}

\newcommand{\bfcase}[2]{\begin{cases}{#1}&\textbf{(bosons)}\\{#2}&\textbf{(fermions)}\end{cases}}
\newcommand{\re}[1]{\mathrm{Re}(#1)}
\newcommand{\im}[1]{\mathrm{Im}(#1)}

\newcommand{\tp}{\intercal}

\pgfplotsset{compat=1.18}

\usepackage{makecell}
\newcommand{\iii}{\mathbb{i}}
\newcommand{\cetimes}{\,\tilde{\times}\,}

\begin{document}

\title{Representation theory of inhomogeneous Gaussian unitaries}

\author{Jingqi Sun}
\email{jingqis6@student.unimelb.edu.au}
\affiliation{School of Mathematics and Statistics, The University of Melbourne, Parkville, VIC 3010, Australia}

\author{Joshua Combes}
\email{josh.combes@unimelb.edu.au}
\affiliation{School of Mathematics and Statistics, The University of Melbourne, Parkville, VIC 3010, Australia}
\affiliation{School of Physics, The University of Melbourne, Parkville, VIC 3010, Australia}

\author{Lucas Hackl}
\email{lucas.hackl@unimelb.edu.au}
\affiliation{School of Mathematics and Statistics, The University of Melbourne, Parkville, VIC 3010, Australia}
\affiliation{School of Physics, The University of Melbourne, Parkville, VIC 3010, Australia}

\begin{abstract}
Gaussian unitaries, generated by quadratic Hamiltonians, are fundamental in quantum optics and continuous-variable computing. Their structures correspond to symplectic (bosons) and orthogonal (fermions) groups, but physical realizations give rise to their respective double covers, introducing phase and sign ambiguities. The homogeneous (quadratic-only) case has been resolved through a parameterization constructed in a recent work [arXiv:2409.11628]. We extend the previous framework to inhomogeneous Gaussian unitaries parameterized by $(M,z,\Psi)$. The Baker-Campbel-Hausdorff formula allows us then to factor any Gaussian unitary into a squeezing and a displacement transformation, from which we derive the group multiplication law.
\end{abstract}

\maketitle

\section{Introduction}\label{sec:intro}
Gaussian states are quantum states whose Wigner functions take a Gaussian form in phase space. They play a foundational role in quantum optics~\cite{gerry2023introductory} and continuous-variable quantum information~\cite{wang2007quantum,weedbrook2012gaussian,adesso2014continuous,holevo2019quantum}. Gaussian states also arise in condensed matter physics~\cite{bogoliubov1947theory,bogoljubov1958new}, quantum field theory in curved spacetime~\cite{wald1994quantum,parker2009quantum,derezinski2013mathematics}, and relativistic quantum information~\cite{Adesso_2012,Friis_2013}. Gaussian unitary transformations, generated by exponentials of quadratic Hamiltonians, preserve the Gaussian character of quantum states. Moreover, such unitaries underpin key protocols including quantum teleportation, squeezing, and Boson Sampling~\cite{weedbrook2012gaussian,aaronson2011computational}. Early works by Holevo and Werner clarified the mathematical structure of Gaussian states and channels~\cite{holevo2012quantum,werner2001gaussian}, while more recent developments have explored their use in quantum simulation, control, and communication~\cite{usenko2025cvqc,q36d-w649}.

Beyond practical applications, recent research has highlighted the geometric and phase properties of Gaussian unitaries, particularly through the lens of the symplectic and metaplectic groups~\cite{Dereziński_Gérard_2023,Woit:2017vqo,Simon1993,folland1989harmonic,degosson2006symplectic}, linked to structures such as the Maslov index and semiclassical analysis~\cite{littlejohn1986semiclassical,cahill1999density}. Previous works derived closed-form expressions of the double cover group multiplication rule using phase space formalism for bosons in~\cite{PhysRevA.61.022306} and for fermions in~\cite{Bravyi_2017}. Notably, Dias and König examine the parametrization and global phases arising in superpositions of Gaussian states~\cite{Dias:2023arXiv230712912D,Dias:2024arXiv240319059D}, although their analysis does not explicitly address the group multiplication law for Gaussian unitaries. More recently, Quesada developed methods for computing vacuum-to-vacuum amplitudes in quadratic bosonic evolutions, offering new analytical techniques for phase characterization in bosonic quantum dynamics~\cite{Quesada2025}.

A system of $N$ bosons can be characterized by their $2N$ linear observables
\begin{align}\label{eq:xi}
    \hat{\xi}^a&\equiv(\underbrace{\hat{q}_1,\cdots,\hat{q}_N}_\text{position},\underbrace{\hat{p}_1,\cdots,\hat{p}_N}_\text{momentum})\,,
\end{align}
also known as quadrature operators, which we collect in a single vector $\hat{\xi}^a$. A general quadratic Hamiltonian (GQH) for bosons takes the form
\begin{align}\label{eq:GQHbosons}
    \hat{H}=\underbrace{\tfrac{1}{2}h_{ab}\hat{\xi}^a\hat{\xi}^b}_\text{quadratic}+\underbrace{f_a\hat{\xi}^a}_\text{linear}+\underbrace{c\id}_\text{scalar}\,.
\end{align}
Its exponential $\mathcal{U}=e^{-\ii\hat{H}}$ and their products for different $\hat{H}$ have the property to map Gaussian states to Gaussian states. They generate the Lie group we call `inhomogeneous' Gaussian unitaries to emphasize that $\hat{H}$ is allowed to include a linear term. The action of such an inhomogeneous Gaussian unitary can be conveniently encoded via symplectic group element $M^a{}_b$ and a displacement vector $z^a$, such that
\begin{align}
    \mathcal{U}^\dagger \hat{\xi}^a\mathcal{U}=M^a{}_b\hat{\xi}^b+z^a\,,
    \label{eq:rep-formula}
\end{align}
which determines $\mathcal{U}$ up to a complex phase. While knowing this phase is not demanded for all physical applications, there are a number of scenarios, such as describing the superposition of Gaussian states~\cite{10.21468/SciPostPhys.17.3.082,Marshall:23,PRXQuantum.2.040315,PhysRevA.47.5024} and controlled Gaussian operations~\cite{PhysRevD.111.105031,PhysRevD.96.085012}, where this phase is crucial. Moreover, this phase is also critical when characterizing the group multiplication law without representing $\mathcal{U}$ as operators acting on an infinite-dimensional Hilbert space. The previous~\cite{Hackl_2024} that we build on characterized `homogeneous' Gaussian unitaries, namely those generated by purely quadratic Hamiltonians (PQH) where the linear term $f_a$ in~\eqref{eq:GQH} vanishes.

In this manuscript, we define the complex phase $\Psi$ of $\mathcal{U}$ with respect to a Gaussian state $\ket{J}$ as $\Psi^*=\braket{J|\mathcal{U}|J}/|\braket{J|\mathcal{U}|J}|$ for bosons and then derive the general group multiplication law
\begin{align}
    \mathcal{U}(M_1,z_1,\Psi_1)\,\mathcal{U}(M_2,z_2,\Psi_2)=\mathcal{U}(M_3,z_3,\Psi_3)\,.
\end{align}
Where $M_3=M_1M_2$ and $z_3=z_1+M_1z_2$ follow readily from the representation theory via~\eqref{eq:rep-formula}. The relation for the phase
\begin{align}
    \Psi_3=\Psi_1\Psi_2 e^{\ii \zeta(M_1,M_2,z_1,z_2)}
\end{align}
is less straightforward and constitutes one of the important results from this manuscript. Here $\zeta$ is the \emph{inhomogeneous cocycle function} which tracks the phases of unitary product. Furthermore, we show how deduce $\Psi$ from $\hat{H}$ when $\mathcal{U}(M,z,\Psi)=e^{-\ii \hat{H}}$. Finally, we treat the equivalent case for fermions, where Gaussian unitaries form are related to the orthogonal matrices $M^a{}_b$, and discuss where $M^a{}_b$ is not connected to the identity due to $\det{M}=-1$.\\

This manuscript is structured as follows: In Section~\ref{sec:method}, we review the relevant techniques and prior work, along with the fundamental ingredients necessary for extending the parametrization. We also elaborate on the motivation, rationale, and conceptual framework for this extension. In Section~\ref{sec:results}, we synthesize these components to derive a complete expression for the inhomogeneous Gaussian unitary transformation. This includes the characterization of the Gaussian unitaries generated by a general quadratic Hamiltonian (GQH), their associated global complex phases, and the full composition rules governing their group multiplication. Section~\ref{sec:cases} presents numerical verifications and visualizations of our results for one bosonic mode, particularly on the evolution of the cocycle function and complex phases generated by Hamiltonians. We conclude in Section~\ref{sec:conc} with a discussion of future directions and potential applications, particularly in the contexts of phase estimation and continuous-variable quantum computing.

\section{Homogeneous Gaussian states and displacements}\label{sec:method}

After introducing our notation, we give a concise review of the parametrization of homogeneous Gaussian unitaries proposed in~\cite{Hackl_2024} and discuss the displacement group. These groups will then be combined in section~\ref{sec:results}.

\subsection{Operators and structures}\label{sec:2:structures}

The notation for operators and structures we develop below correspond directly to experimental objects like beam splitters, squeezers, and displacements.

We consider a bosonic or fermionic system with $N$ modes. Classically, the system is described by a phase space $V\simeq\mathbb{R}^{2N}$ and its dual space $V^*$. We use abstract index notation~\cite[appendix~A]{hackl2020geometry} to denote vectors by $v^a\in V$ and dual vectors by $w_a\in V^*$. The quantum phase space can be constructed from the linear observables~\eqref{eq:xi} given by the operator-valued vector $\hat{\xi}^a\equiv(\hat{q}_1,\cdots,\hat{q}_N,\hat{p}_1,\cdots,\hat{p}_N)$. These are known as quadratures satisfying the canonical commutation relations $[\hat{q}_i,\hat{p}_j]=\ii\delta_{ij}$ or as Majorana modes satisfying the canonical anti-commutation relations $\{\hat{q}_i,\hat{p}_j\}=\ii\delta_{ij}$ for fermions. These relations can be combined into $\Lambda^{ab}$ with inverse $\lambda_{ab}$
\begin{align}\label{eq:com-og}
    \Lambda^{ab}=\bfcase{\hspace{1.4mm}[\hat{\xi}^a,\hat{\xi}^b]=\ii\Omega^{ab}}{\{\hat{\xi}^a,\hat{\xi}^b\}=G^{ab}}\,,
\end{align}
where $\Omega$ is a symplectic form and $G$ a positive-definite metric whose matrix representations in above basis is
\begin{align}\label{eq:def:matrix}
    \Omega^{ab}\equiv\begin{pmatrix}0&\id\\-\id&0\end{pmatrix}\quad\text{and}\quad G^{ab}\equiv\begin{pmatrix}\id&0\\0&\id\end{pmatrix}\,.
\end{align}
Their inverses are denoted by $\omega_{ab}$ and $g_{ab}$, respectively, with $\Omega^{ac}\omega_{cb}=G^{ac}g_{cb}=\delta^a{}_b$. For future use, we also introduce the symbol
\begin{align}
    \iii=\bfcase{\ii}{1}\,.
\end{align}
More details are in table~\ref{tab:matrices} and appendix~\ref{ap:A:operator}.

The most general quadratic Hamiltonian (GQH) is
\begin{align}\label{eq:GQH}
    \hat{H}=\bfcase{\frac{1}{2}h_{ab}\hat{\xi}^a\hat{\xi}^b+f_a\hat{\xi}^a+c\id}{\frac{\ii}{2}h_{ab}\hat{\xi}^a\hat{\xi}^b+c\id},
\end{align}
where $h_{ab}$ in above basis must be real and symmetric for bosons, while real and antisymmetric for fermions. In addition, there is an imaginary multiplier $\ii$ on the quadratic term for fermions. The linear term, determined by the dual vector $f_a$, only exists for bosons, as treating the same case for fermions will either break the fermionic super selection rule (parity) or require $f_a$ to be Grassmann-valued, which would not be physical.

We refer to a Hamiltonian of the form~\eqref{eq:GQH} with $f=0$ and $c=0$ as purely quadratic Hamiltonian (PQH). We will see that they form a faithful representation of the symplectic or special orthogonal Lie algebra for bosons and fermions, respectively. We define these Lie algebras and their associated Lie groups in their fundamental representation acting on phase space as
\begin{subequations}
\begin{align}
    \mathfrak{g}&=\left\{K\in\mathfrak{gl}(2N,\mathbb{R})\,\big|\,K\Lambda=-\Lambda K^\tp\right\}\,,\\
    \mathcal{G}&=\left\{M\in\mathrm{GL}(2N,\mathbb{R})\,\big|\,M\Lambda M^\intercal=\Lambda\right\}\,,
\end{align}
\end{subequations}
summarized in table~\ref{tab:structures}.

\begin{table}[t]
    \centering
    \renewcommand{\arraystretch}{1.5}
\begin{tabular}{c c c} 
\hline
\hline
 \textbf{structures} & \textbf{bosons} & \textbf{fermions} \\
 \hline
 Lie algebras $\mathfrak{g}$ & $\mathfrak{sp}(2N,\mathbb{R})$ & $\mathfrak{so}(2N,\mathbb{R})$ \\
 Lie groups $\mathcal{G}$ & $\mathrm{Sp}(2N,\mathbb{R})$ & $\mathrm{O}(2N,\mathbb{R})$ \\
 double covers $\widetilde{\mathcal{G}}$ & $\mathrm{Mp}(2N,\mathbb{R})$ & $\mathrm{Pin}(2N,\mathbb{R})$ \\
\hline
tensor $\Lambda^{ab}$ & $\ii\Omega^{ab}=[\hat{\xi}^a,\hat{\xi}^b]$ & $G^{ab}=\{\hat{\xi}^a,\hat{\xi}^b\}$ \\
    $\lambda_{ab}=(\Lambda^{-1})_{ab}$ & $-\ii\omega_{ab}$ & $g_{ab}$ \\
    $\iii$ & $\ii$ & $1$ \\
\hline
\hline
    \textbf{operators} & \multicolumn{2}{c}{\textbf{unified notation}} \\
    \hline
    linear (vector) & \multicolumn{2}{c}{$\widehat{z}=-z^a\lambda_{ab}\hat{\xi}^b$} \\
    linear (co-vector) & \multicolumn{2}{c}{$\widehat{w}=-\iii w_a\hat{\xi}^a$} \\
    quadratic & \multicolumn{2}{c}{$\widehat{K}=\frac{1}{2}\lambda_{ac}K^c{}_b\hat{\xi}^a\hat{\xi}^b$} \\
\hline
\hline
\end{tabular}
\ccaption{Algebraic structures and boson-fermion unified notation}{The Lie algebras $\mathfrak{g}$ that quadratic Hamiltonians $\hat{H}$ form, corresponding Lie groups $\mathcal{G}$ and the double covers $\widetilde{\mathcal{G}}$ that $\exp(-\ii\hat{H})$ form. Unified notation allows the most of derivations of linear and quadratic operators to be merged into a consistent form for both bosons and fermions.}
\label{tab:structures}
\renewcommand{\arraystretch}{1}
\end{table}

\subsection{Homogeneous Gaussian unitaries}\label{sec:2:hom-gu}

Here, we summarize the key results of~\cite{Hackl_2024,Hackl_2021} concerning \emph{homogeneous} Gaussian unitaries, generated by PHQ. This will be required to generalize the construction to inhomogeneous unitaries, generated by GQH, which also include displacements.

\subsubsection{Squeezing operators}

In the following, we show that PQHs form a faithful representation of the Lie algebra $\mathfrak{g}$. However, their exponentials realize only a projective representation of the Lie group $\mathcal{G}$ because of phase ambiguities in the group composition. To obtain a true (linear) representation, we must instead consider the double cover group $\widetilde{\mathcal{G}}$.

We consider the space generated by PQHs in the form
\begin{align}\label{eq:Khat}
    \widehat{K}=-\ii\hat{H}=\tfrac{1}{2}\lambda_{ac}K^c{}_b\hat{\xi}^a\hat{\xi}^b\,,
\end{align}
where $K^c{}_b=\frac{1}{\iii}\Lambda^{ca}h_{ab}$ lies in the respective Lie algebra $\mathfrak{g}$. Using~\eqref{eq:com-og} and BCH, one finds the relation
\begin{align}
    [\widehat{K}_1,\widehat{K}_2]&=\widehat{[K_1,K_2]}\,,\\
    e^{-\widehat{K}}\hat{\xi}^ae^{\widehat{K}}&=(e^{K})^a{}_b\hat{\xi}^b\,.
\end{align}
The first equation demonstrates that $\widehat{K}$ forms a Lie algebra
representation, while the second line shows that all $e^{\widehat{K}}$ generate a (potentially projective) representation of the Lie group $\mathcal{G}$.

\textbf{Why we need the double cover:} For a single bosonic or fermionic degree of freedom, we consider
\begin{align}
    K=\begin{pmatrix}
        0 & 1\\
        -1 & 0
    \end{pmatrix},\,\, \widehat{K}=\bfcase{-\frac{\ii}{2}(\hat{q}^2+\hat{p}^2)}{+\hat{q}\hat{p}}
\end{align}
corresponding to $\widehat{K}=-\ii(\hat{n}\pm \frac{1}{2})$, where $\hat{n}$ has spectrum $\mathbb{Z}_{\geq0}$ for bosons ($+$) and $\{0,1\}$ for fermions ($-$). To connect the algebra element $K$ with its group action, we exponentiate both generators. This yields the classical symplectic transformation and its quantum counterpart
\begin{align}
   \hspace{-2.5mm} M(t)\!=\!e^{tK}\!=\!\begin{pmatrix}
    \cos{t} & \sin{t}\\
    -\!\sin{t} & \cos{t}
    \end{pmatrix},\, \mathcal{S}(t)\!=\!e^{t\widehat{K}}\!=\!e^{-\ii t (\hat{n}\pm \frac{1}{2})}.
\end{align}
Here $\mathcal{S}(t) = e^{t\widehat{K}}$ denotes the unitary operator generated by the quadratic Hamiltonian $\widehat{K}$, representing the quantum analogue of the classical symplectic transformation $M(t) = e^{tK}$. While $M(t)$ is a rotation in the plane and clearly $2\pi$-periodic, we find $\mathcal{S}(2\pi)=-\id$, so this unitary operator is $4\pi$-periodic.

The same argument also works for several degrees of freedom~\cite{Hackl_2024} and establishes that for each group element $M\in\mathcal{G}$, there exist two unitaries $\pm\mathcal{S}(M)$ in the set generated by all $e^{\widehat{K}}$ with
\begin{align}\label{eq:fix-S(M)-up-to-sign}
    \mathcal{S}^\dagger(M)\hat{\xi}^a\mathcal{S}(M)=M^a{}_b\hat{\xi}^b\,.
\end{align}
Note that this equation does not fix the complex phase or sign of $\mathcal{S}(M)$.

\subsubsection{Reference complex phase and Gaussian states}

We have just seen that $M\in\mathcal{G}$ in~\eqref{eq:fix-S(M)-up-to-sign} determines a given unitary only up to a complex phase. In order to compare different such unitaries, we somehow need to assign a reference complex phase. We will do this by fixing a reference quantum state $\ket{J}$ and then associating to each Gaussian unitary $\mathcal{U}$ the complex phase of the expectation value $\braket{J|\mathcal{U}|J}$, provided it does not vanish.

We choose this reference quantum state to be a so-called Gaussian state that is fully determined by a complex structure $J: V\to V$ satisfying the following properties: (i) $J^2=-\id$, (ii) $J\Lambda J^\intercal=\Lambda$ and (iii) for bosons we require $-J\Omega$ to be positive-definite. These conditions can be encoded in the 2-out-of-3 property via
\begin{align}\label{eq:Kaehler-relation}
    G\omega=-J\quad\Leftrightarrow\quad J=\Omega g\,,
\end{align}
where we introduce the missing structure $G$ for bosons and $\Omega$ for fermions by solving~\eqref{eq:Kaehler-relation}. The conditions (i)-(iii) for $J$ to define a Gaussian state $\ket{J}$ is then equivalent to requiring that $G$ is a positive-definite metric for bosons and $\Omega$ a symplectic form for fermions, respectively. The triple $(\Omega,G,J)$ is called Kähler structure and there always exists a basis, such that their matrix forms are
\begin{align}\label{eq:Kaehler-standard-forms}
    \Omega=\left(\begin{array}{c|c}
       0  &  \id\\
         \hline
       -\id  & 0
    \end{array}\right)\,,\,G=\left(\begin{array}{c|c}
       \id  &  0\\
         \hline
       0  & \id
    \end{array}\right)\,,\,J=\left(\begin{array}{c|c}
       0  &  \id\\
         \hline
       -\id  & 0
    \end{array}\right)\,.
\end{align}

Given such $J$, we define $\ket{J}$ as any normalized solution to the equation
\begin{align}\label{eq:annihilation-op}
    \tfrac{1}{2}(\delta^a{}_b+\ii J^a{}_b)\hat{\xi}^b\ket{J}=0\,,
\end{align}
which one can show to determine $\ket{J}$ uniquely up to a complex phase~\cite{hackl2021bosonic}. For any operator $\mathcal{A}$ we can then introduce the \emph{reference complex phase}
\begin{align}
    \Phi_J[\mathcal{A}]=\frac{\braket{J|\mathcal{A}|J}}{|\braket{J|\mathcal{A}|J}|}\,,
\end{align}
provided $\braket{J|\mathcal{A}|J}\neq 0$. Note that this is invariant under changing the complex phase of $\ket{J}$.

With respect to $J$, we can always decompose a linear map $K: V\to V$ into its linear part $C_M$ and anti-linear part $D_M$, as
\begin{subequations}
\begin{align}
    C_M&=\tfrac{1}{2}(M-JMJ)\,,\\
    D_M&=\tfrac{1}{2}(M+JMJ)\,,
\end{align}
\end{subequations}
such that $M=C_M+D_M$ and satisfies $[C_M,J]=0$ and $\{D_M,J\}=0$. If $C_M$ is invertible, we can also define
\begin{align}
    Z_M=C_M^{-1}D_M=\frac{\id-\Delta_{M^{-1}}}{\id+\Delta_{M^{-1}}}\quad\!\!\text{with}\!\!\quad \Delta_M=-MJM^{-1}J\,.
\end{align}

Given $K^a{}_b$ with $[K,J]=0$, we further define
\begin{align}\label{eq:def-trbar-detbar}
    \overline{\Tr}(K)=\Tr(\overline{K})\quad\text{and}\quad\overline{\det}(K)=\det(\overline{K})\,,
\end{align}
where we decompose $K=K_1+\ii K_2$ as
\begin{align}
    K\equiv\left(\begin{array}{c|c}
    K_1     & K_2 \\
         \hline
    -K_2     & K_1
    \end{array}\right)
\end{align}
in an arbitrary basis, where $J$ takes the standard form~\eqref{eq:Kaehler-standard-forms}. While $K_1$ and $K_2$ depend on the specific basis,~\eqref{eq:def-trbar-detbar} is basis-independent~\cite{Hackl_2024}.

\subsubsection{Circle function and cocycles}

Result~1 of~\cite{Hackl_2024} showed that for a given $\mathcal{S}(M)$, we have
\begin{align}\label{eq:old-result}
    \braket{J|\mathcal{S}(M)|J}^2=\overline{\det}\left(\frac{M-JMJ}{2}\right)^{\iii^2}\,,
\end{align}
This can be rephrased as $\Phi_J[\mathcal{S}(M)]^2=\varphi(M)$, where we introduce the circle function $\varphi: \mathcal{G}\to\mathrm{U}(1)$ with
\begin{align}\label{eq:circle-function}
    \varphi(M)=\frac{\overline{\det}(C_M)}{|\overline{\det}(C_M)|}\,.
\end{align}

We can now compare the phase $\Phi_J(\mathcal{S}(M_1M_2))$ with $\Phi_J(\mathcal{S}(M_1))$ and $\Phi_J(\mathcal{S}(M_2))$ by first looking at their respective squares. This motivates the relation
\begin{align}
    e^{\ii \eta(M_1,M_2)}=\frac{\Phi^2_J(\mathcal{S}(M_1M_2))}{\Phi^2_J(\mathcal{S}(M_1))\Phi^2_J(\mathcal{S}(M_2))}=\frac{\varphi(M_1M_2)}{\varphi(M_1)\varphi(M_2)}\,,
\end{align}
which defines the cocycle function $\eta: \mathcal{G}\times\mathcal{G}\to \mathrm{U}(1)$ uniquely by requiring $\eta$ to be continuous and satisfy $\eta(\id,\id)=0$. It can be computed explicitly as
\begin{align}\label{eq:cocycle-function}
    \eta(M_1,M_2)&=\mathrm{Im}\,\overline{\mathrm{Tr}}\log\big(\id-Z_{M_1}Z_{M_2^{-1}}\big)\,.
\end{align}

\subsubsection{Parametrization}

The above facts imply that the squeezing operators actually form a representation of the double cover $\widetilde{\mathcal{G}}$ of the Lie group $\mathcal{G}$. This representation is either infinite-dimensional (bosons) or $2^N$-dimensional (fermions), so it is desirable to describe the multiplication law in $\widetilde{\mathcal{G}}$ using smaller matrices, such as the elements of $\mathcal{G}$.

The idea of~\cite{Hackl_2024} building on~\cite{rawnsley2012universal} is to construct the double cover as the subset
\begin{align}\label{group:DC}
    \widetilde{G}=\{(M,\psi)\mid\varphi(M)=\psi^2\}\subset\mathcal{G}\times\mathrm{U}(1)\,,
\end{align}
with the circle function $\varphi$ from~\eqref{eq:circle-function}. For each $M\in\mathcal{G}$, there exist thus two elements $(M,\pm\sqrt{\varphi(M)})$ in the double cover with the natural projection $\pi: \widetilde{G}\to\mathcal{G}$, such that $\pi(M,\psi)=M$.

It is shown in~\cite{Hackl_2024} that the group multiplication can be efficiently described as
\begin{align}\label{eq:group-mult}
    (M_1,\psi_1)\cdot(M_2,\psi_2)=(M_1\cdot M_2,\psi_1\psi_2e^{\frac{\ii}{2}\eta(M_1,M_2)})\,.
\end{align}

The operator $S(M,\psi)$ represents the uniquely defined group element $(M,\psi)$ in the double cover $\widetilde{\mathcal{G}}$, implementing the symplectic transformation $M$ while fixing its global phase relative to the reference state $\ket{J}$. The squeezing operator $\mathcal{S}(M,\psi)$ is uniquely characterized by requiring
\begin{subequations}\label{eq:def:sq}
\begin{align}
    \mathcal{S}^\dagger(M,\psi)\hat{\xi}^a\mathcal{S}(M,\psi)&=M^a{}_b\hat{\xi}^b\,,\label{eq:def:sq-psi-op}
    \\
    \Phi_J[\mathcal{S}(M,\psi)]&=\psi^*\,.\label{eq:def:psi}
\end{align}
\end{subequations}
As a proper representation, we consequently have
\begin{align}\label{eq:sq-gm}
    \mathcal{S}(M_1,\psi_1)\cdot\mathcal{S}(M_2,\psi_2)=\mathcal{S}(M_1\cdot M_2,\psi_1\psi_2e^{\frac{\ii}{2}\eta(M_1,M_2)})\,.
\end{align}

\subsubsection{Cartan decomposition}

Tracking the phase effect of composing displacements and squeezings is the prerequisite to deduct the complex phase and multiplication laws of an inhomogeneous Gaussian unitary. Given a Gaussian state $\ket{J}$, the Cartan decomposition allows us to decompose any group element as $M=Tu$, where only $u$ contributes to the phase. The components $T=(MM^{\mathsf T})^{1/2}$ is the squeezing and $u:=T^{-1}M$ is a passive rotation (orthogonal-symplectic).

For a Lie group element $M\in\mathcal{G}$ and $K\in\mathfrak{g}$, the Cartan decomposition $M=Tu$ is given by~\cite{Hackl_2021,Hackl_2024} 
\begin{align}
    T&=\sqrt{\Delta_M}\quad\text{and}\quad u=T^{-1}M\,,
\end{align}
where $JT=T^{-1}J$ and $Ju=uJ$. One can show~[CITE] that $\varphi(T)=1$, so that we can choose
\begin{align}
    \mathcal{S}(M,\psi)=\mathcal{S}(T,1)\mathcal{S}(u,\psi)\,,
\end{align}
where we used $\eta(\id,u)=0$. This shows
\begin{align}\label{eq:sq_phase}
    \Phi_J[\mathcal{S}(M,\psi)]=\Phi_J[\mathcal{S}(u,\psi)]\,.
\end{align}
In physical terms, $\mathcal{S}(u,\psi)$ is a passive transformation satisfying $\mathcal{S}(u,\psi)\ket{J}=\psi\ket{J}$, while $\mathcal{S}(T,u)$ is a pure squeezing transformation that does not change the phase.

\subsubsection{Gaussian unitary from PQH}

We have seen that $M\in\mathcal{G}$ does not suffice to determine the complex phase of a squeezing operator $\mathcal{S}(M)$, which is why introduced $\mathcal{S}(M,\psi)$. However, given a Lie algebra element $K$, we have the isomorphism~\eqref{eq:Khat} that uniquely specifies $\widehat{K}$ giving rise to a squeezing operator $\exp(\widehat{K})$. It is therefore natural to ask what is the complex phase $\psi=\Psi_J[e^{\widehat{K}}]$, such that $\mathcal{S}(e^K,\psi)=e^{\widehat{K}}$.

This question was answered for bosons and fermions in Result 3a of~\cite{Hackl_2024}, which we compactly review in appendix~\ref{ap:psi}. The strategy was to decompose $\widehat{K}=\widehat{K}_I+\widehat{K}'$, such that $K_I$ has purely imaginary eigenvalues and a commuting remainder $K'$. By choosing an adapted reference complex structure $\tilde{J}$, the argument (phase) of the expectation value $\braket{J|\exp{\widehat{K}}|J}$ is expressed as a sum of three contributions: (i) a simple trace term $\Tr(K_I\tilde{J})$ determined by the eigenvalues of $K_I$, (ii) a cocycle term $\eta$ that encodes the change of basis relating $J$ to $\tilde{J}$, and (iii) a determinant term involving the remainder $K'$. Thus, this closed expression isolates the otherwise ambiguous global phase in terms of classical symplectic data and the cocycle function.

\emph{Note that the bosonic case was also considered more recently in~\cite{Quesada2025}, where $\braket{J|e^{-\ii\hat{H}}|J}$ was referred to as vacuum-to-vacuum amplitude.}

\subsection{Displacements}

The class of Gaussian unitaries generated by linear observables is known as the displacement operators, denoted by $\mathcal{D}$. They form the displacement group, which we want to combine with the \emph{homogeneous} Gaussian unitaries discussed in section~\ref{sec:2:hom-gu}.

\subsubsection{Displacement operators}\label{sec:2:disp}

Detailed derivations of the subsection can be found in appendix~\ref{ap:B:lin-com}. Given a phase space vector $z^a\in V$ or a dual vector $w_a\in V^*$, we define
\begin{align}
    \widehat{z}=-z^a\lambda_{ab}\hat{\xi}^b\quad\text{and}\quad \widehat{w}=-\iii w_a\hat{\xi}^a\,.
\end{align}
We define the isomorphism between $V$ and $V^*$ given by $\iii \lambda_{ab}$ that takes a vector $z^a$ and maps it to the dual vector $\tilde{z}_a=\iii\lambda_{ab}z^b$ with inverse $\iii \Lambda^{ab}$ that maps a dual vector $w_a$ to $\tilde{w}^a=\iii\Lambda^{ab}w_b$. We then have $\widehat{z}=\widehat{w}$ for $w_a=\tilde{z}_a$ or $\tilde{w}^a=z^a$. Here, we list their commutators as
\begin{subequations}
\begin{align}
    [\widehat{z}_1,\widehat{z}_2]&=-z_1\lambda z_2=z_2\lambda z_1\,,\label{eq:com:z}
    \\
    [\widehat{w}_1,\widehat{w}_2]&=-w_1\Lambda w_2=w_2\Lambda w_1\,,
    \\
    [\widehat{Fz},\widehat{z}]&=[\widehat{w},\widehat{wF}]\,,\label{eq:fzfw}
\end{align}
\end{subequations}
for any linear map $F^a{}_b$. For the bosonic case, the line~\eqref{eq:com:z} gives a symplectic structure $\omega(\cdot,\cdot)$ as
\begin{align}\label{eq:symp_form}
    [\widehat{z_1},\widehat{z_2}]=\ii z_1\omega z_2=\ii\omega(z_1,z_2)\,.
\end{align}
We further find the relation
\begin{align}
    [\hat{\xi}^a,\widehat{z}]&=z^a\,,\label{eq:com-xi-z}
    \\
    e^{-\widehat{z}}\hat{\xi}^ae^{\widehat{z}}&=\hat{\xi}^a+z^a\,,\label{eq:expz_gen}
\end{align}
which justifies to define the displacement operator as
\begin{align}\label{eq:def:disp}
    \mathcal{D}(z)=e^{\widehat{z}}\quad\text{with}\quad\mathcal{D}^\dagger(z)\hat{\xi}^a\mathcal{D}(z)=\hat{\xi}^a+z^a\,,
\end{align}
satisfying the relation product relation
\begin{align}\label{eq:displacement-multiplication}
    \mathcal{D}(z_1)\mathcal{D}(z_2)&=e^{\frac{1}{2}[\widehat{z}_1,\widehat{z}_2]}\mathcal{D}(z_1+z_2)\,,
\end{align}
which we will be able to relate to the Heisenberg group.

\subsubsection{Displacement and Heisenberg group}

Using the multiplication law of displacement operators~\eqref{eq:displacement-multiplication}, we can define the displacement group
\begin{align}
    \mathcal{D}_N=\mathbb{R}^{2N}\times\mathrm{U}(1)\label{eq:group:H}
\end{align}
with the multiplication law
\begin{align}\label{eq:displacement-multiplication-law}
    \hspace{-2mm}(z_1,e^{\ii\vartheta_1})\cdot(z_2,e^{\ii\vartheta_2})=(z_1+z_2,e^{\ii(\vartheta_1+\vartheta_2+\frac{1}{2}\omega(z_1,z_2))}).
\end{align}
By construction the displacement operators $\mathcal{D}(z,e^{\ii\vartheta})=e^{\ii\vartheta}e^{\widehat{z}}$ with additional phase form a faithful representation of the displacement group.

This group is closely related to the real Heisenberg group $\mathbb{H}_N$, which is diffeomorphic to $\mathbb{R}^{2N+1}$ and satisfies the multiplication law
\begin{align}
    (z_1,\vartheta_1)\cdot(z_2,\vartheta_2)=(z_1+z_2,\vartheta_1+\vartheta_2+\tfrac{1}{2}\omega(z_1,z_2))\,.
\end{align}
$\mathbb{H}_N$ forms the universal cover of $\mathcal{D}_N$ via the mapping
\begin{align}\label{eq:covering-map}
    \mathbb{H}_N\ni(z,\vartheta)\mapsto(z,e^{\ii\vartheta})\in\mathcal{D}_N\,,
\end{align}
which implies that displacement operators can also be understood as (non-faithful) representation of the Heisenberg group. Here, $\mathrm{U}(1)$ describes the complex phase factor $e^{\frac{1}{2}[\widehat{z}_1,\widehat{z}_2]}$, thereby yielding a central extension from the Abelian translation group $\mathbb{R}^N$, which arises from the non-commutativity of quantum observables.

\section{Inhomogeneous Gaussian unitaries}\label{sec:results}

Our aim here is to write down the rule for how squeezing and displacement operations combine for bosons, including the correct overall phase. We do this by working in a lifted space where each operation carries a phase tag $\Psi$, so that composition is associative and matches operator multiplication. This lets us predict, \eg measurable phase effects in Gaussian circuits and do phase-correct gate compilation for bosonic codes.

From a more mathematical perspective, we want to extend the double cover group $\widetilde{\mathcal{G}}$ to its inhomogeneous version $\mathrm{I}\widetilde{\mathcal{G}}$. If we consider a squeezing operator followed by displacement operator while ignoring the sign ambiguity as $(M,z)$ schematically, \cite{Hackl_2021} also provides their multiplication rules
\begin{align}\label{eq:Mz-ml}
    (M_1,z_1)\cdot(M_2,z_2)&=(M_1\cdot M_2,z_1+M_1z_2)\,.
\end{align}
Hence our main goal is to solve the completed case where we consider $(M,z,\Psi)$ for some overall function $\Psi(M,z)$, and come up with their multiplication rules.

Our first result will be an explicit product law for the complex phase. First, we recall the BCH composition rules for displacement and squeezing operators. Next, we apply the Lie-algebra inverse to express any quadratic Hamiltonian as a unitary with an added phase $\theta$. Finally, we compute the extra phase $\zeta(M_1,M_2,z_1,z_2)$ that appears when multiplying two Gaussian unitaries. This pins down the inhomogeneous group for bosonic systems as the double cover
\begin{align}
    \mathrm{IMp}(2N,\mathbb{R})=(\mathrm{Sp}(2N,\mathbb{R})\ltimes\mathbb{R}^{2N})\cetimes\mathrm{U}(1)\,.
\end{align}
with a central $\mathrm{U}(1)$ phase.

\subsection{Multiplication law}

An inhomogeneous operation $\mathcal{U}$ is parametrized by $(M,z,\Psi)$ where $M^a{}_b$ is a squeezing matrix, $z^a$ a displacement vector, such that
\begin{align}\label{eq:rep-inh}
    \mathcal{U}^\dagger\hat{\xi}^a\mathcal{U}=M^a{}_b\hat{\xi}^b+z^a\quad\text{and}\quad\Phi_J[\mathcal{U}]=\Psi^*\,.
\end{align}
This only becomes a complete parametrization when one specifies a unique multiplication law.

Recall for a Lie group element $M\in\mathcal{G}$, we can always decompose $M=Tu$ with $T=\sqrt{\Delta_M}$. Then the Cartan decomposition gives
\begin{align}
    \mathcal{S}(M,\psi)=\mathcal{S}(T,1)\mathcal{S}(u,\psi)\,,
\end{align}
with
\begin{align}
    \Phi_J[\mathcal{S}(M,\psi)]=\Phi_J[\mathcal{S}(u,\psi)]=\psi^*\,.
\end{align}

\begin{lemma}\label{lemma:gamma}
\textbf{Expectation value of $\mathcal{D}(z)\mathcal{S}(T,1)$ and its phase $\gamma$.} 
The other decomposed squeezing operator would merge with the displacement operator and have the following expectation value
\begin{align}\begin{split}\label{eq:fullvalue}
    &\braket{J|\mathcal{D}(z)\mathcal{S}(T,1)|J}\\&=\det(\id-Y_M^2)^{1/8}e^{-\frac{1}{4}zg(\id+Y_M)z}e^{\frac{\ii}{4}z\omega Y_Mz}\,,   
\end{split}\end{align}
where
\begin{align}
    Y_M&=\frac{\id-\Delta_M}{\id+\Delta_M}=\frac{2}{\id+\Delta_M}-\id=Z_{M^{-1}}\,,
    \\
    \Delta_{M}&=-MJM^{-1}J=MM^\tp\,.
\end{align}
For convenience, we also denote the phase as
\begin{align}\label{eq:phase_gamma}
    \Phi_J[\mathcal{D}(z)\mathcal{S}(T,1)]&=e^{\ii\gamma(M,z)}\,,
\end{align}
where the displaced-squeezing phase is
\begin{align}\label{eq:gamma}
    \gamma(M,z)&=\frac{\ii}{4}[\widehat{Y_Mz},\widehat{z}]=\frac{1}{4}z\omega Y_Mz=\frac{1}{4}\omega(z,Y_Mz)\,,
\end{align}
and its modulus as
\begin{align}\label{eq:evmodulus}
    |\braket{J|\mathcal{D}(z)\mathcal{S}(T,1)|J}|=|\braket{e^{\widehat{K}_+}}|e^{-\frac{1}{2}zg(\id+MM^\tp)^{-1}z}\,,
\end{align}
where $\braket{e^{\widehat{K}_+}}$ is formulated in~\eqref{eq:<eK+>}.
\end{lemma}

\begin{proof}
Basically, one can use relevant results from~\cite[section~III-A-3]{Hackl_2021} to make a further derivation to find the explicit expectation value $\braket{0|e^{\widehat{z}}e^{\widehat{K}_+}|0}$. In particular, formulas~\cite[Eq.~(155,156,159,162)]{Hackl_2021} provide the normal ordering of displacement, purely squeezing and their interplays. Please see appendix~\ref{ap:gamma} for details, where a unified derivation is also made for fermions.
\end{proof}

The next step is to appropriately decompose a general Gaussian unitary into displacements and squeezings.

\begin{lemma}\label{lemma:decomp}
\textbf{Decomposition of unitary.} 
The most general inhomogeneous Gaussian unitary $\mathcal{U}$ can be decomposed as
\begin{subequations}
\begin{align}
    \mathcal{U}(M,z,\Psi)&=e^{\ii\theta}\mathcal{D}(z)\mathcal{S}(M,\psi)\,,\label{eq:unitary-decomp}\\
    e^{\ii\theta}&=\Psi^*\psi e^{-\ii\gamma(M,z)}\,,\label{eq:eitheta}
\end{align}
\end{subequations}
where $\psi=\pm\sqrt{\varphi(M)}$. This decomposition is unique up to the choice of $\psi$, which can also be made unique by fixing a convention, such as $\psi=\sqrt{\varphi(M)}$, where the square-root is defined as the principal branch of the complex square root with $\sqrt{-1}=\ii$ rather than $-\ii$.
\end{lemma}

\begin{proof}
First, let us prove that $(M,z,\Psi)$ together with the relations~\eqref{eq:rep-inh} characterizes the unitary operator $\mathcal{U}(M,z,\Psi)$ uniquely. It follows from the Proposition 6 from~\cite{hackl2021bosonic}, that any complex phase can be absorbed into the phase $e^{\ii\theta}$, while $\theta$ is only unique up to a $2\pi$ periodicity. \\
Second, we plug the decomposition~\eqref{eq:unitary-decomp} into the first equation~\eqref{eq:rep-inh} use the properties of the displacement operator~\eqref{eq:def:disp} and the squeezing operator~\eqref{eq:def:sq-psi-op}, which we find to be saturated, as shown in~\eqref{eq:proof:uxiu}. \\
Finally, we put the decomposition~\eqref{eq:unitary-decomp} into the second equation of~\eqref{eq:rep-inh} yielding
\begin{align}\begin{split}\label{eq:argpsi}
    \Psi^*&=\Phi_J[e^{\ii\theta}\mathcal{D}(z)\mathcal{S}(M,\psi)]\\
    &=e^{\ii\theta}\Phi_J[\mathcal{D}(z)\mathcal{S}(T,1)\mathcal{S}(u,\psi)]\\
    &=e^{\ii\theta}\Phi_J[\mathcal{D}(z)\mathcal{S}(T,1)]\psi^*\\
    &=\psi^*e^{\ii(\theta+\gamma(M,z))}\,,
\end{split}\end{align}
where we use~\eqref{eq:def:psi}, the Cartan decomposition~\eqref{eq:sq_phase} as $\mathcal{S}(u,\psi)\ket{J}=\psi\ket{J}$, and Lemma~\ref{lemma:gamma}. The last line directly implies~\eqref{eq:eitheta}.\\
Note that the choice of $\psi=\pm\sqrt{\varphi(M)}$ does not affect this proof, but will flip the sign of $e^{\ii\theta}$ accordingly.
\end{proof}

Before we come up with the multiplication rules, we need to appropriately commute the displacement and squeezing operators.

\begin{lemma}\label{lemma:commute}
\textbf{Commutation law.} 
Given a squeezing operator $\mathcal{S}(M,\psi)$ and a displacement operator $\mathcal{D}(z)$, they commute in the way of a semi-direct product
    \begin{align}
    \mathcal{S}(M,\psi)\mathcal{D}(z)=\mathcal{D}(Mz)\mathcal{S}(M,\psi)\,.
\end{align}
\end{lemma}

\begin{proof}
    From the definition~\eqref{eq:def:sq-psi-op}, we have for both bosons and fermions
    \begin{align}\begin{split}
    &\mathcal{S}(M,\psi)\mathcal{D}(z)\mathcal{S}^\dagger(M,\psi)\\
    &=\mathcal{S}^\dagger(M^{-1},\psi^*)e^{-z^a\lambda_{ab}\hat{\xi}^b}\mathcal{S}(M^{-1},\psi^*)\\
    &=e^{-z^a\lambda_{ab}(M^{-1})^b{}_c\hat{\xi}^c}\\
    &=e^{-(Mz)^b\lambda_{bc}\hat{\xi}^c}\\
    &=\mathcal{D}(Mz)\,.
    \end{split}\end{align}
    Adding after both side a squeezing $\mathcal{S}(M,\psi)$ gives the desired commutation rule for the displacement and squeezing.
\end{proof}

Our next goal is to figure out the actual multiplication rules for the two unitaries.

\begin{result1}
Inhomogeneous Gaussian unitaries from a representation $\mathcal{U}:\mathrm{I}\widetilde{\mathcal{G}}\to\mathrm{Lin}(\mathcal{H})$ that characterized by the conditions
\begin{subequations}\label{eq:result1}
\begin{align}
    \mathcal{U}^\dagger(M,z,\Psi)\hat{\xi}^a\mathcal{U}(M,z,\Psi)&=M^a{}_b\hat{\xi}^b+z^a\,, \label{eq:result1_a}
    \\
    \Phi_J[\mathcal{U}(M,z,\Psi)]&=\Psi^*\,,\label{eq:result1_b}
\end{align}
\end{subequations}
where $M^a{}_b\in\mathcal{G}$ is a matrix and $z^a\in V$ is a vector in the phase space. They forms a unitary representation satisfying the multiplication rule
\begin{align}\begin{split}\label{eq:multiplication-law}
    &\mathcal{U}(M_1,z_1,\Psi_1)\cdot\mathcal{U}(M_2,z_2,\Psi_2)
    \\&=\mathcal{U}(M_1M_2,z_1+M_1z_2,\Psi_1\Psi_2e^{\ii\zeta(M_1,M_2,z_1,z_2)})\,,
\end{split}\end{align}
where we call
\begin{align}\begin{split}\label{eq:def:zeta}
    \zeta(M_1,M_2,z_1,z_2)={}&\tfrac{1}{2}\mathrm{Im}\,\overline{\mathrm{Tr}}\log\big(\id-Y_{M_1^{-1}}Y_{M_2}\big)\\
    &+\gamma(M_1,z_1)+\gamma(M_2,z_2)\\
    &-\gamma(M_1M_2,z_1+M_1z_2)\\
    &-\tfrac{1}{2}\omega(z_1,M_1z_2)\,,
\end{split}\end{align}
the inhomogeneous cocycle function and $\gamma$ was defined in~\eqref{eq:gamma}.
\end{result1}

\begin{proof}
We plug the decomposition from Lemma~\ref{lemma:decomp} for both unitaries into the LHS of~\eqref{eq:multiplication-law} and compute
\begin{align}\begin{split}\label{eq:dec-mult}
    &\mathcal{U}(M_1,z_1,\Psi_1)\cdot\mathcal{U}(M_2,z_2,\Psi_2)\\
    &=e^{\ii(\theta_1+\theta_2)}\mathcal{D}(z_1)\mathcal{D}(M_1z_2)\mathcal{S}(M_1,\psi_1)\mathcal{S}(M_2,\psi_2)\\
    &=e^{\ii(\theta_1+\theta_2)}e^{\frac{1}{2}[\widehat{z_1},\widehat{M_1z_2}]}\mathcal{D}(z_1+M_1z_2)\\
    &\qquad\times\mathcal{S}(M_1M_2,\psi_1\psi_2e^{\frac{\ii}{2}\eta(M_1,M_2)})\\
    &=e^{\ii \theta}\mathcal{D}(z)\mathcal{S}(M,\psi)\,,
\end{split}\end{align}
where we recognized in the final expression again the decomposition from Lemma~\ref{lemma:decomp} with
\begin{subequations}
\begin{align}
    M&=M_1M_2\,,\\
    z&=z_1+M_1z_2\,,\\
    \psi&=\psi_1\psi_2e^{\frac{\ii}{2}\eta(M_1,M_2)}\,,\\
    \theta&=\theta_1+\theta_2+\tfrac{1}{2}\omega(z_1,M_1z_2)\,,\label{eq:r1:eitheta}
\end{align}
\end{subequations}
where we use the symplectic form in~\eqref{eq:r1:eitheta} as from~\eqref{eq:symp_form}. The only remaining piece from this is to compute $\Psi$ via
\begin{align}
    \Psi^*&=\Phi_J[e^{\ii \theta}\mathcal{D}(z)\mathcal{S}(M,\psi)]\,,
\end{align}
which we already did in Lemma~\ref{lemma:decomp} to find
\begin{align}
\begin{split}
    \Psi&=\psi e^{-\ii(\theta+\gamma(M,z))}\\
    &=\Psi_1e^{\ii(\theta_1+\gamma(M_1,z_1))}\Psi_2e^{\ii(\theta_2+\gamma(M_2,z_2))}e^{\frac{\ii}{2}\eta(M_1,M_2)}
    \\&\qquad\times e^{-\ii(\theta_1+\theta_2+\frac{1}{2}\omega(z_1,M_1z_2)+\gamma(M_1M_2,z_1+M_1z_2))}\\
    &=\Psi_1\Psi_2e^{\ii(\gamma(M_1,z_1)+\gamma(M_2,z_2)+\frac{1}{2}\eta(M_1,M_2))}
    \\&\qquad\quad\times e^{-\ii(\frac{1}{2}\omega(z_1,M_1z_2)+\gamma(M_1M_2,z_1+M_1z_2))}\\
    &=\Psi_1\Psi_2e^{\ii\zeta(M_1,M_2,z_1,z_2)}\,,
\end{split}
\end{align}
where we use the fact from~\eqref{eq:argpsi} that $\psi=\Psi e^{\ii(\theta+\gamma(M,z))}$. The last line complete the proof for~\eqref{eq:def:zeta}.
\end{proof}

To summarize the idea, each component of $\zeta$ arises from decomposing and combining parts of displacement and squeezing operators. We mark the generated processes for each terms in their braces, where the arrow $\to$ represents the rewriting
\begin{align}\begin{split}
    \overbrace{\zeta(M_1,M_2,z_1,z_2)}^{\mathcal{U}_1\mathcal{U}_2\to\mathcal{U}_3}&=
    \overbrace{\gamma(M_1,z_1)}^{\mathcal{U}_1\to\mathcal{D}_1\mathcal{S}_1}+\overbrace{\gamma(M_2,z_2)}^{\mathcal{U}_2\to\mathcal{D}_2\mathcal{S}_2}
    \\
    &\quad\underbrace{-\tfrac{1}{2}\omega(z_1,M_1z_2)}_{\mathcal{D}_1\mathcal{D}_2'\to\mathcal{D}_3}+\underbrace{\tfrac{1}{2}\eta(M_1,M_2)}_{\mathcal{S}_1\mathcal{S}_2\to\mathcal{S}_3}
    \\
    &\quad\underbrace{-\gamma(M_1M_2,z_1+M_1z_2)}_{\mathcal{D}_3\mathcal{S}_3\to\mathcal{U}_3}\,,
\end{split}\end{align}
where the dashed displacement $\mathcal{D}'$ is a result of commutation $\mathcal{S}_1\mathcal{D}_2\to\mathcal{D}_2'\mathcal{S}_1$, which is also the source for the semi-direct product of the group. 

One another intention is to find the set for the inhomogeneous group $\mathrm{I}\widetilde{\mathcal{G}}$. Similar to the Poincaré group $\mathrm{SO}(1,3)\ltimes\mathbb{R}^{1,3}$, we would expect the resulting group to be a semi-direct product between metapletic group and the Heisenberg group. The group structure can be deducted from (\ref{eq:group:H}) and (\ref{group:DC}), gives
\begin{align}\begin{split}
    \mathrm{IMp}(2N,\mathbb{R})&=\mathrm{Sp}(2N,\mathbb{R})\ltimes\mathbb{D}_N
    \\&\cong\mathrm{ISp}(2N,\mathbb{R})\cetimes\mathrm{U}(1)
    \\&=(\mathrm{Sp}(2N,\mathbb{R})\ltimes\mathbb{R}^{2N})\cetimes\mathrm{U}(1)
    \\&=\{(M,z,\Psi)\}\,.
\end{split}\end{align}
The squeezing phases $\pm\psi\in\mathrm{U}(1)$ is now combined with the displacement phase group $\mathrm{U}(1)$ into a new overall phase group $\mathrm{U}(1)$ with elements $\Psi$. Notice that $\cetimes$ is referred as a twisted product, meaning the parametrization of inhomogeneous metaplectic group is a central $\mathrm{U}(1)$-extension of the inhomogeneous symplectic group $\mathrm{Sp}(2N,\mathbb{R})\ltimes\mathbb{R}^{2N}$.

\subsection{Gaussian unitary from GQH}

In some cases, a general quadratic Hamiltonian $\hat{H}$ is given, which will generate an inhomogeneous Gaussian unitary transformation $\mathcal{U}=e^{-\ii\hat{H}}$. One would than like to know how to deduct the unitary entries $(M,z,\Psi)$ from the known Hamiltonian parameters $(h,f,c)$. This primarily deals with determining $\Phi_J(e^{-\ii\hat{H}})$.

This task require an explicit relation between linear and quadratic operators. Consider the famous Baker-Campbell-Hausdorff (BCH) relation between the exponentials of $\widehat{w}=-\iii w_a\hat{\xi}^a$ and $\widehat{K}=\frac{1}{2}\lambda_{ac}K^c{}_b\hat{\xi}^a\hat{\xi}^b$. Note that we will allow for $K\in\mathfrak{g}_{\mathbb{C}}$ and $w\in V^*_{\mathbb{C}}$, \ie the resulting operators may not be anti-Hermitian, such that their exponentials may not be unitary.

\begin{lemma}\label{lemma:BCH:exp(w)exp(K)}
The BCH relation gives the following operator expression
\begin{subequations}
\begin{align}
    e^{\widehat{K}}e^{\widehat{w}}&=\exp\big(\widehat{K}+\widehat{w\alpha(K)}+[\widehat{w},\widehat{w\beta(K)}]\big)\,,\label{eq:BCH-lin-qu}
    \\
    e^{\widehat{w}}e^{\widehat{K}}&=\exp\big(\widehat{K}+\widehat{w\alpha(-K)}+[\widehat{w},\widehat{w\beta(K)}]\big)\,,\label{eq:BCH-qu-lin}
\end{align}
\end{subequations}
where the following objects are introduced
\begin{subequations}
\begin{align}
    \alpha(K)&=\frac{K}{e^{K}-\id}\,,\label{eq:ab}\\
    \beta(K)&=\frac{1}{4}\frac{K-\sinh{K}}{\id-\cosh{K}}\,.\label{eq:beta}
\end{align}
\end{subequations}
The inverse is
\begin{subequations}
\begin{align}
    w_a&=f_b\alpha^{-1}(-K)^b{}_a\,,
    \\
    e^{\widehat{K}+\widehat{f}}&=e^{-[\widehat{w},\widehat{w\beta(K)}]}e^{\widehat{w}}e^{\widehat{K}}\,.\label{eq:inverse}
\end{align}
\end{subequations}
\end{lemma}

\begin{proof}
Their commutators are given by
\begin{subequations}
\begin{align}
    [\widehat{w},\widehat{K}]&=\widehat{wK}\,,
    \\
    [\widehat{w},[\widehat{K}]^n]&=\widehat{wK^n}\,.
\end{align}
\end{subequations}
More commutators as discussed in appendix~\ref{ap:B:nested} will also be applied. The key method involves solving the differential equations of Lie algebras
\begin{align}
    e^{-X(t)}\frac{\mathrm{d}}{\mathrm{d}t}(e^{X(t)})=\sum_{n=0}^\infty\frac{[\dot{X}(t),[X(t)]^n]}{(n+1)!}\,,
\end{align}
where $X(t)$ is a smooth, matrix-valued function with derivative $\dot{X}(t)=\frac{\mathrm{d}X(t)}{\mathrm{d}t}$. Then comparing with the series generated by nested commutators between $\hat{w}$ and $\hat{K}$, with each quadratic, linear and scalar term, would archive each desired formula. Please check appendix~\ref{ap:ab} for the full proof in details.

For the inverse, it is obviously the quadratic terms match as $\widehat{K}$. The linear term is
\begin{align}
    \widehat{w\alpha(-K)}=-\ii f_a\hat{\xi}^a=\widehat{f}\,,
\end{align}
gives $w\alpha(-K)=f$ and thus
\begin{align}
    w_a=f_b\alpha^{-1}(-K)^b{}_a=f_b\left(\frac{\id-e^{-K}}{K}\right)^b_{\hspace{2mm}a}\,.
\end{align}
The linear term just follows from the commutator.
\end{proof}

We would now consider the bosonic case only, to avoid the subtle Grasmann property for fermionic systems.

\begin{result2}
Given a bosonic GQH $\hat{H}$ with parameters $(h,f,c)$ as in~\eqref{eq:GQH}, we have $e^{-\ii\hat{H}}=\mathcal{U}(M,z,\Psi)$ given by
\begin{align}
    M^a{}_b&=(e^K)^a{}_b\quad\text{with}\quad K^a{}_b=\Omega^{ac}h_{cb}\,,\\
    z^a(h,f)&=f_c(e^{-\Omega h}-\id)^c{}_b(h^{-1})^{ba}\,,\\
    \Psi&=\Phi_J^*[e^{\widehat{K}}]e^{\ii c}e^{-\ii\frac{1}{4}z\omega\Sigma(K)z}\,,\\
    \Sigma(K)&=\frac{2}{\id-e^KJe^{-K}J}-\id+\frac{K-\sinh{K}}{\id-\cosh{K}}\,.
\end{align}
where $\beta(K)$ is the function introduced in~\eqref{eq:beta} and $\Phi_J[e^{\widehat{K}}]$ can be computed by~\eqref{eq:main-bosons} based on Result~3a of~\cite{Hackl_2024}, as reviewed in appendix~\ref{ap:psi}.
\end{result2}

\begin{proof}
We start with the decomposition $\mathcal{U}(M,z,\Psi)=e^{\ii\theta}\mathcal{D}(z)\mathcal{S}(M,\psi)$ from~\eqref{eq:unitary-decomp} and set it equal to $e^{-\ii \hat{H}}$ for a GQH~\eqref{eq:GQH} to solve for $(M,z,\Psi)$. We make the bosonic setting as
\begin{align}
    -{\tfrac{\ii}{2}}\omega_{ac}K^c{}_b\hat{\xi}^a\hat{\xi}^b=\widehat{K}=-{\textstyle\frac{\ii}{2}}h_{ab}\hat{\xi}^a\hat{\xi}^b\,,
\end{align}
hence $\omega_{ac}K^c{}_b=h_{ab}$ or $K^a{}_b=\Omega^{ac}h_{cb}$ as usual. Using the inverse~\eqref{eq:inverse} in Lemma~\ref{lemma:BCH:exp(w)exp(K)}, we find
\begin{align}\begin{split}
    \mathcal{U}=e^{-\ii\hat{H}}&=e^{\widehat{K}+\widehat{f}-\ii c\id}
    \\&=e^{-\ii c-[\widehat{\beta(K)z},\widehat{z}]}e^{\widehat{z}}e^{\widehat{K}}
    \\&=e^{\ii(z\omega\beta(K)z-c)}\mathcal{D}(z)\mathcal{S}(e^K,\psi)\,,
\end{split}\end{align}
where we use the fact that for bosons
\begin{align}\begin{split}
    z^a&=\Omega^{ab}w_b=\Omega^{ab}f_c\left(\frac{\id-e^{-K}}{K}\right)^c_{\hspace{2mm}b}\\&=f_c\left(\frac{e^{-K}-\id}{K}\right)^c_{\hspace{2mm}b}\Omega^{ba}\\&=f_c(e^{-\Omega h}-\id)^c{}_b(h^{-1})^{ba}\,,
\end{split}\end{align}
and~\eqref{eq:fzfw} which gives
\begin{align}
    [\widehat{w},\widehat{w\beta(K)}]=[\widehat{\beta(K)z},\widehat{z}]=-\ii z^a\omega_{ab}\beta(K)^b{}_cz^c\,.
\end{align}
To align with the previous decomposition~\eqref{eq:unitary-decomp}, we have $M=e^K$, $\mathcal{D}(z)$ identical and $\theta=z\omega\beta(K)z-c$. Notice that
\begin{align}\begin{split}
    \gamma(M,z)+\theta+c&=\gamma(e^K,z)+z\omega\beta(K)z
    \\&=\frac{1}{4}z\omega Y_{e^K}z+z\omega\frac{1}{4}\frac{K-\sinh{K}}{\id-\cosh{K}}z
    \\&=\frac{1}{4}z\omega\Sigma(K)z\,,
\end{split}\end{align}
where
\begin{align}
    \Sigma(K)=\frac{2}{\id-e^KJe^{-K}J}-\id+\frac{K-\sinh{K}}{\id-\cosh{K}}\,.
\end{align}
Relation~\eqref{eq:eitheta} state that
\begin{align}\begin{split}
    \Psi=\psi e^{-\ii(\gamma(M,z)+\theta)}=\psi e^{\ii c}e^{-\ii\frac{1}{4}z\omega\Sigma(K)z}\,.
\end{split}\end{align}
while how to determine $\psi=\Phi_J[e^{\widehat{K}}]$ is explained in details in appendix~\ref{ap:psi}. 
\end{proof}

\subsection{Addendum: Fermionic Gaussian unitaries}\label{sec:fermions}

The Lie group $\mathcal{G}=\mathrm{Sp}(2N,\mathbb{R})$ for bosons is connected and hence completely generated by (\ref{eq:def:sq}). However, the actual Lie group for fermions is $\mathcal{G}=\mathrm{O}(2N,\mathbb{R})$, which consists of two disconnected components. The subset $\mathrm{SO}(2N,\mathbb{R})$ is connected to the identity $\id$ with $\det(M)=1$ for $M\in\mathrm{SO}(2N,\mathbb{R})$, which can be generated by (\ref{eq:def:sq}) as well. To include the other component with elements satisfy $\det(M)=-1$, we need an alternative form of squeezing operators as
\begin{align}
    \mathcal{S}(M_w)=w_a\hat{\xi}^a\,,\qquad\qquad\textbf{(fermions)}
\end{align}
where neither $w_a\in V^a$ and $\hat{\xi}^a$ is Grassmann-valued. Here the induced group elements are
\begin{align}\label{eq:sq_mw}
    (M_w)^a{}_b=w_cG^{ca}w_b-\delta^a{}_b\,,
\end{align}
with the normalization condition $w_aG^{ab}w_b=2$. Please check details in appendix~\ref{ap:fermion-sq}.

Although not every elements in the other component can be written as $M_w$, $M_wM$ where $M\in\mathrm{SO}(2N,\mathbb{R})$ would solve the problem. That is, given a starting point $w_a$, we can represent all elements in the other component by $M_wM$. Hence the actual form of the squeezing operators corresponding to it would be $\mathcal{S}(M_wM)$. Hence a consistent way to define the complex phase of the case would be
\begin{align}
    \Phi_J[\mathcal{S}(M_w)\mathcal{S}(M_wM,\psi)]=\psi\,,
\end{align}
such that it would have exactly the same phase $\psi$ as $\mathcal{S}(M,\psi)$ for $M\in\mathrm{SO}(2N,\mathbb{R})$.

\begin{result3}
As a summary of explanation, if $\det(M)=1$, then it is just in the $\mathrm{SO}(2N,\mathbb{R})$ group, \ie the part connected to the identity. Otherwise, we need to write it in the product $M_wM$, which is now in the $\mathrm{SO}(2N,\mathbb{R})$ group. We may define a new phase as:
\begin{align}
    \Phi_{J,w}=\Phi_J(M_w^{(1-\det(M))/2}M)\,,
\end{align}
to track the actual phase. In general, previously we have the $\mathrm{Pin}(2N,\mathbb{R})\subset\mathrm{O}(2N,\mathbb{R})\times\mathrm{U}(1)$.
\end{result3}

\section{Case studies for one bosonic mode}\label{sec:cases}

We provide numerical validations and illustrative examples of our results for quantum systems with a single bosonic degree of freedom, representing the simplest nontrivial case. We further examine how the logarithm of a specific class of Gaussian unitaries can lead to non-quadratic operators.

\subsection{Review: Parametrization and Stability}

To facilitate the analysis and provide clear illustrations in the case of a single bosonic mode, we parametrize the symplectic algebra and the corresponding quadratic Hamiltonians.

\textbf{Lie algebra.} The symplectic Lie algebra $\mathrm{sp}(2,\mathbb{R})$ is isomorphic to $\mathrm{sl}(2,\mathbb{R})$ and generated by the three matrices
\begin{align}\label{eq:case-study-XYZ-matrices}
    X=\begin{pmatrix}
    0 & 1\\
    1 & 0
    \end{pmatrix}\,,\quad Y=\begin{pmatrix}
    1 & 0\\
    0 & -1
    \end{pmatrix}\,,\quad Z=\begin{pmatrix}
    0 & 1\\
    -1 & 0
    \end{pmatrix}\,.
\end{align}
\textbf{Representation theory.} For a bosonic system with $N=1$, these three generators are represented as anti-Hermitian quadratic operators via relation~\eqref{eq:Khat}, which yields
\begin{subequations}\label{eq:case-study-XYZ-operators}
\begin{align}
    2\ii\widehat{X}&=\hat{p}^2-\hat{q}^2=-\hat{a}^{\dagger 2}-\hat{a}^2\,,\\
    2\ii\widehat{Y}&=\hat{p}\hat{q}+\hat{q}\hat{p}=\ii(\hat{a}^{\dagger 2}-\hat{a}^2)\,,\\
    2\ii\widehat{Z}&=\hat{p}^2+\hat{q}^2=2\hat{n}+1\,,
\end{align}
\end{subequations}
where $\hat{n}=\frac{1}{2}(\hat{p}^2+\hat{q}^2-1)$ is a number operator, whose spectrum consists of non-negative integers.

With respect to this basis, a general element of the Lie algebra, $K=aX+bY+cZ$, is represented by the quadratic operator $\widehat{K}=a\widehat{X}+b\widehat{Y}+c\widehat{Z}$. We define that $r=\sqrt{a^2+b^2-c^2}$. A straightforward calculation shows that $\det(K)=-r^2$, and that the eigenvalues of $K$ are $\pm r$.

In this representation, $\widehat{X}$ and $\widehat{Y}$ generate squeezing operations that differ by a $45^\circ$ rotation in phase space, whereas $\widehat{Z}$ generates phase rotations. Taking into account this symmetry, together with the algebraic structure encoded in $r=\sqrt{a^2+b^2-c^2}$, it is natural to reduce the number of independent parameters by setting either $a=0$ or $b=0$, since both correspond to squeezing generators with the same algebraic form. Without loss of generality, we restrict our attention to the subspace $K=aX+cZ$.

The \emph{stability} of bosonic systems has been discussed in detail in ~\cite{PhysRevA.97.032321}. In the single-mode setting, stability is determined by the eigenvalues of $K$ given by $\pm r$:
\begin{itemize}
    \item \textbf{Stable (imaginary $r$).} The standard example is $\widehat{H}=-\frac{\ii}{2}(\hat{p}^2+\hat{q}^2)$, which means the evolution follows a Harmonic oscillator.
    \item \textbf{Unstable (real $r$).} The standard example is $\widehat{H}=-\frac{\ii}{2}(\hat{p}^2-\hat{q}^2)$, which means the evolution follows an inverted harmonic oscillator. This is clearly unstable and generates squeezing.
    \item \textbf{Metastable ($r=0$, but $K\neq0$).} The standard example is $\widehat{H}=-\frac{\ii}{2}\hat{p}^2$ describing a free propagating particle. The corresponding generator $K$ is non-diagonalizable, meaning there exists a non-trivial Jordan-block in it.
\end{itemize}

\subsection{Verification of the cocycle function}\label{sec:4:time}

In this case study, we numerically verify Result~1. The analytical expression for the inhomogeneous cocycle function $\zeta(M_1,M_2,z_1,z_2)$ can be tested numerically within a truncated Hilbert space. In addition, we perform the same analysis for the expectation value of the total number operator $\braket{\hat{N}_J}$ in order to illustrate the dynamical properties of randomly generated Hamiltonians.

We consider a single bosonic mode with Hilbert space spanned by the Fock basis $\mathcal{H}=\{\ket{n}:n=0,1,2,\dots\}$, where $\ket{n}$ denotes the number (Fock) state with $n$ excitations. A truncated Hilbert space is introduced by imposing a cutoff excitation number $n_{\mathrm{max}}$, restricting the system to the subspace $\{\ket{n}:n=0,1,\dots,n_\mathrm{max}\}$. For sufficiently low mean excitation number, $\braket{\hat{n}} \ll n_{\mathrm{max}}$, this truncation induces negligible errors. In particular, this allows the operators $\hat{\xi}^a$ to be represented by finite-dimensional square matrices of size $n_{\mathrm{max}}+1$. As $n_{\mathrm{max}}$ increases, both the simulation accuracy and the computational cost increase accordingly.

Given a randomly-generated generalnquadratic Hamiltonian (GQH) $\hat{H}_\ell$, the corresponding time-dependent inhomogeneous Gaussian unitary is given by
\begin{align}
    \mathcal{U}_\ell(h_\ell,f_\ell;t)=e^{-\ii \hat{H}_\ell t}=e^{\widehat{\Omega h_\ell}t+\widehat{f}t}\,,
\end{align}
with randomly-generated parameters $(h_\ell,f_\ell)$. The scalar generator $c_\ell$ is omitted here, since the associated global phase cancels out in both evolutions. We demonstrate the time evolution of the cocycle $\zeta(t)$ and of the total number expectation value $\braket{\hat{N}_J(t)}$ for two independently generated sets of Hamiltonian parameters, $(h_1,f_1)$ and $(h_2,f_2)$. In all truncated numerical computations, the reference state $\ket{J}$ is approximated by the vacuum state $\ket{0}$.

\begin{figure*}[t]
    \centering
    \begin{tikzpicture}
    \begin{scope}[xshift=3mm]
    \draw (-4.65,0) node[inner sep=0pt]{\includegraphics[width=1.13\columnwidth]{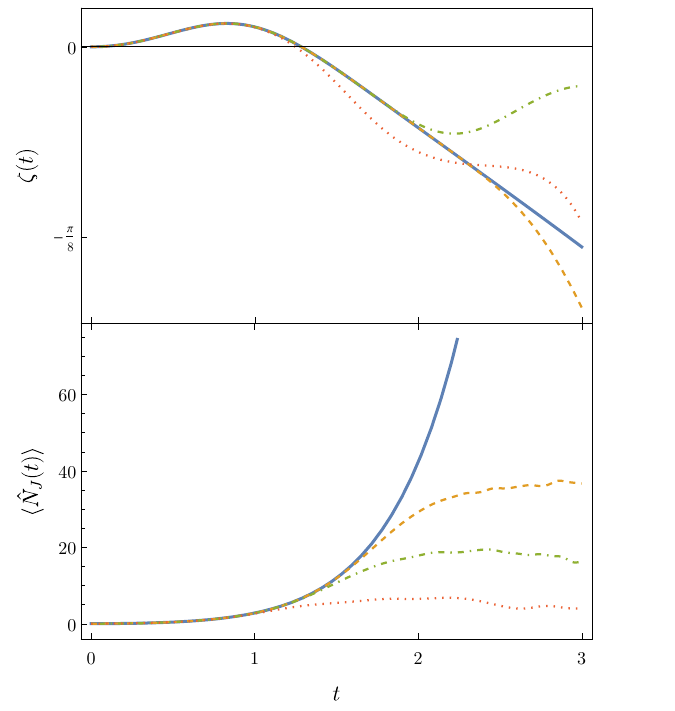}};
    \draw (-4.65,5.5) node{(a) unstable example};
    \end{scope}
    \draw (4.65,0) node[inner sep=0pt]{\includegraphics[width=1.13\columnwidth]{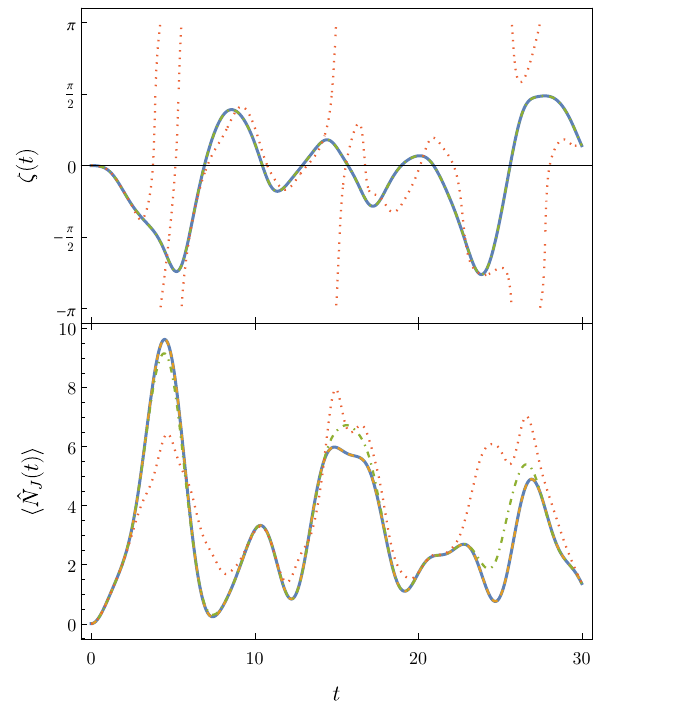}};
    \draw (4.65,5.5) node{(b) stable example};
    \draw (-6.4,-1) node[inner sep=0pt]{\includegraphics[width=0.27\columnwidth]{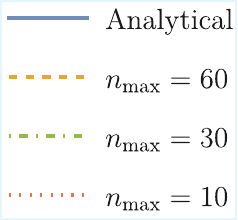}};
    \end{tikzpicture}
    \ccaption{Time evolution of the product phase $\zeta(t)$ and $\braket{\hat{N}_J(t)}$}{The complex phase $\zeta \in [-\pi,\pi]$ and the expectation value of the total number operator $\braket{\hat{N}_J(t)}$ are plotted as functions of time $t$ for randomly generated Hamiltonians with parameters $(h_\ell,f_\ell)$. Solid lines denote the analytical results given by~\eqref{eq:def:zeta} and~\eqref{eq:def:NJ}, while dashed lines represent truncated numerical approximations with different cutoff values $n_{\mathrm{max}}$, computed using~\eqref{eq:num:zeta} and~\eqref{eq:NJ-time}. All panels share a common legend shown in the lower-left panel. The unstable case has $h_1=\left(\begin{smallmatrix}0.4&-0.6\\-0.6&0.5\end{smallmatrix}\right)$ and $h_2=\left(\begin{smallmatrix}0.5&0.4\\0.4&-0.4\end{smallmatrix}\right)$ with real eigenvalues $\pm 0.4$ and $\pm 0.6$, respectively. Whereas the stable case has $h_1=\left(\begin{smallmatrix}0.4&0.2\\0.2&0.5\end{smallmatrix}\right)$ and $h_2=\left(\begin{smallmatrix}0.8&-0.2\\-0.2&0.5\end{smallmatrix}\right)$ with purely imaginary eigenvalues $\pm 0.4\ii$ and $\pm 0.6\ii$, respectively. In all cases, the displacement parameters are fixed to $f_1=f_2=(0.5,0.5)$.}
    \label{fig:time-evolution}
\end{figure*}

\textbf{The inhomogeneous cocycle function $\zeta$.} Our numerical simulation of~\eqref{eq:def:zeta} requires a concrete operational definition. The global phase argument of bosonic systems is defined as $\arg{\braket{J|\mathcal{U}_1|J}}=\arg\Psi_1^*$, which introduces an additional overall minus sign compared with the fermionic convention. As a result, the expression for $\zeta$ to produce the numerical simulation is given by
\begin{align}\label{eq:num:zeta}
    \zeta(t)=\arg\left(\frac{\Phi_J[\mathcal{U}_1(t)]\Phi_J[\mathcal{U}_2(t)]}{\Phi_J[\mathcal{U}_1(t)\mathcal{U}_2(t)]}\right)\,,
\end{align}
where $\ket{J}=\ket{0}$.

\textbf{The total number operator $\braket{\hat{N}_J}$.} The expectation value of the total number operator is defined as
\begin{align}
    \braket{\hat{N}_J(t)}&=\braket{J|\mathcal{U}_2^\dagger(t)\mathcal{U}_1^\dagger(t)\hat{N}_J\mathcal{U}_1(t)\mathcal{U}_2(t)|J}\,,\label{eq:NJ-time}
    \\&=\braket{\tilde{J}(t),\tilde{z}(t)|\hat{N}_J|\tilde{J}(t),\tilde{z}(t)}\,,
\end{align}
where the time-dependence is generated by the Gaussian unitaries $\mathcal{U}_\ell(t)$. With the same identification $\ket{J}=\ket{0}$, the first line~\eqref{eq:NJ-time} is used for the numerical comparison. The corresponding analytical expression can be found in~\cite[Eq. (106)]{Hackl_2021}, and is given by
\begin{align}\label{eq:def:NJ}
    \braket{\hat{N}_J(t)}=-\frac{1}{4}\tr(1-\Delta_{M(t)})+\frac{1}{2}\tilde{z}_a(t)g_{ab}\tilde{z}_b(t)\,,
\end{align}
where
\begin{subequations}
\begin{align}
    M(t)&=M_1(t)M_2(t)=e^{\Omega h_1t}e^{\Omega h_2t}\,,
    \\
    \tilde{z}(t)&=z_1(t)+M_1(t)z_2(t)\,,
    \\
    z_\ell(t)&=f_\ell(e^{-\Omega h_\ell t}-\id)h^{-1}_\ell\,.
\end{align}
\end{subequations}

\textbf{Numerical comparison (figure~\ref{fig:time-evolution}).}
Figure~\ref{fig:time-evolution} shows representative examples of both stable and unstable dynamics. The stable regime requires that both $h_1$ and $h_2$ have positive determinants, ensuring that the eigenvalues of $K_\ell=\Omega h_\ell$ are purely imaginary. In this case, the dynamics are periodic, exhibiting clear cyclic behavior without divergence. By contrast, if either $h_1$ or $h_2$ possesses non-negligible real eigenvalues, the system becomes unstable and displays exponential growth. Consequently, for randomly generated pairs $(h_1,h_2)$, unstable dynamics are statistically more prevalent.

The unstable example behaves as expected: all truncated simulations initially agree with the analytical solution, but deviations appear earlier for lower truncation levels. Increasing the cutoff improves accuracy and delays the onset of deviation. The total number operator grows without bound as a function of time, accounting for the rapid divergence observed in the truncated evolution.

In contrast, the stable example exhibits pronounced cyclic behavior in both the cocycle $\zeta$ and the total number expectation value. As a result, no divergence is observed, and even a relatively low truncation level suffices to reproduce the analytical solution with high accuracy over long times. This agreement provides strong numerical confirmation of Result~1, particularly in the stable regime.

\subsection{Phase of Hamiltonian-generated Gaussian unitary}\label{sec:4:phase}

In this subsection, we study the expectation value $\braket{J|e^{\widehat{K}+\widehat{f}}|J}$ for a single bosonic mode, which can be viewed as the inhomogeneous analogue of the homogeneous case illustrated in~\cite[Fig.~6]{Hackl_2024}. This quantity naturally decomposes into two components: namely the phase $\arg\braket{J|e^{\widehat{K}+\widehat{f}}|J}$ and the modulus $|\braket{J|e^{\widehat{K}+\widehat{f}}|J}|$.

\begin{figure*}[t]
    \centering
    \begin{tikzpicture}
    \draw (0,0) node[inner sep=0pt]{\includegraphics[width=\linewidth]{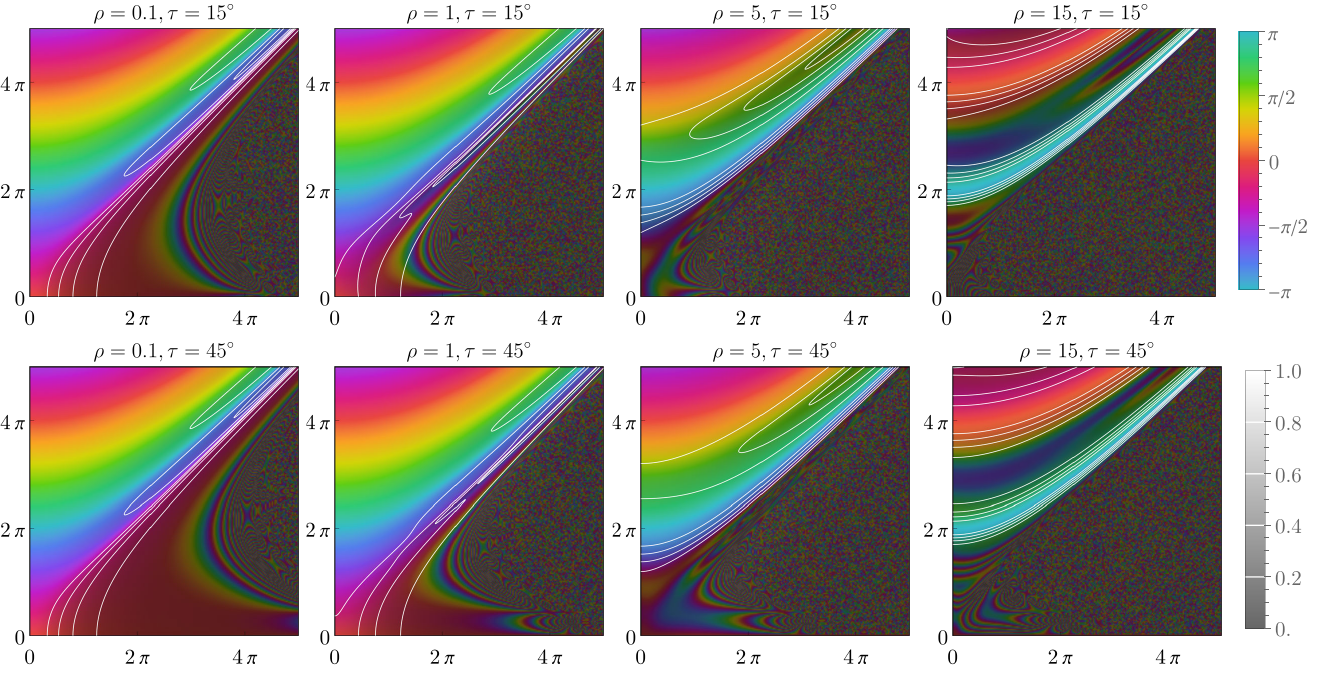}};
    \draw[->, thick] (-8.9,-4.55) -- (-8.9,-3.35) node[right] {$c$};
    \draw[->, thick] (-8.9,-4.55) -- (-7.7,-4.55) node[above] {$a$};
    \draw (8.0,4.5) node[font=\footnotesize]{$\arg(\braket{J|e^{\widehat{K}}|J})$};
    \draw (8.1,-4.3) node[font=\footnotesize]{$|\braket{J|e^{\widehat{K}}|J}|$};
    \end{tikzpicture}
    \ccaption{Bosonic $\braket{J|e^{\widehat{K}+\widehat{f}}|J}$ for $\widehat{K}=a\widehat{X}+c\widehat{Z}$ and $f=\rho(\cos\tau,\sin\tau)$}{These panels present the inhomogeneous case in exactly the same configuration as~\cite[Fig.~6]{Hackl_2024}. We plot $\braket{J|e^{\widehat{K}+\widehat{f}}|J}$ for displacement parameters $\tau=15^\circ,45^\circ$ and $\rho=0.1,1,5,15$. The complex phase~\eqref{eq:arg-ekf} is encoded by color, while the modulus~\eqref{eq:mod-ekf} is represented by brightness. White contour lines indicate respective values of the modulus, as specified by the brightness legend. In all panels, the horizontal axis corresponds to $a$ and the vertical axis to $c$.}
    \label{fig:exp(k+f)}
\end{figure*}

\textbf{Complex phases.} The phase can be extracted directly from Result~2 in the special case $c=0$, yielding
\begin{align}\label{eq:arg-ekf}
    \arg\braket{J|e^{\widehat{K}+\widehat{f}}|J}&=\arg\braket{J|e^{\widehat{K}}|J}+\tfrac{1}{4}z\omega\Sigma(K)z\,,
\end{align}
where
\begin{subequations}
\begin{align}
    z^a&=f_c\left(\frac{e^{-K}-\id}{K}\right)^c_{\hspace{2mm}b}\Omega^{ba}\,,
    \\
    \Sigma(K)&=\frac{2}{\id+e^Ke^{K^\tp}}-\id+\frac{K-\sinh{K}}{\id-\cosh{K}}\,.
\end{align}
\end{subequations}
The evaluation of $\arg\braket{J|e^{\widehat{K}}|J}$ is discussed in detail in appendix~\ref{ap:psi}. 

\textbf{Modulus.} The modulus is obtained by combining Lemma~\ref{lemma:gamma} with~\eqref{eq:old-result}. The central equation~\eqref{eq:evmodulus} of Lemma~\ref{lemma:gamma} implies that the interplay between displacement and squeezing produces an additional factor of the form
\begin{align}
    \frac{|\braket{J|e^{\widehat{z}}e^{\widehat{K}_+}|J}|}{|\braket{J|e^{\widehat{K}_+}|J}|}=e^{-\frac{1}{2}zg(\id+e^Ke^{K^\tp})^{-1}z}\,,
\end{align}
Since~\eqref{eq:old-result} provides the expression for $\braket{J|e^{\widehat{K}}|J}{}^2$, we finally obtain
\begin{align}\label{eq:mod-ekf}
    \braket{J|e^{\widehat{z}}e^{\widehat{K}}|J}{}^2=\frac{e^{-zg(\id+e^Ke^{K^\tp})^{-1}z}}{\overline{\det}\left(\frac{e^K-Je^KJ}{2}\right)}\,.
\end{align}

For a single bosonic mode, the quadratic generator $K$ is characterized by three real parameters $(a,b,c)$ and the linear term $f$ by two additional parameters. As discussed above, one parameter of $K$ can be eliminated by setting $b=0$, leaving the reduced form $K=aX+cZ$. To facilitate a direct comparison with the homogeneous case studied in~\cite[Fig.~6]{Hackl_2024}, we adopt the same configuration and plot results in the $(a,c)$ plane.

The displacement vector $f$ is conveniently parameterized in polar coordinates as
\begin{align}
    f^a=(x,y)=\rho(\cos\tau,\sin\tau)\,.
\end{align}
In practice, we consider only a discrete set of values for $\rho$ and $\tau$. Intuitively, the magnitude $\rho$ of the displacement dominates the effect on the expectation value, while the angular parameter $\tau$, due to its periodicity, plays a comparatively minor role. Accordingly, we present panels for several representative values of $\rho$ and $\tau$. With this setup, the inhomogeneous contribution can be directly involved within the homogeneous result: one simply adds the term $\frac{1}{4}z\omega\Sigma(K)z$ to the phase $\arg\braket{J|e^{\widehat{K}}|J}$ in~\cite[Eq.~(197)]{Hackl_2024}, and multiplies the modulus $|\braket{J|e^{\widehat{K}}|J}|$ in~\cite[Eq.~(196)]{Hackl_2024} by the factor $e^{-\frac{1}{2}zg(\id+MM^\tp)^{-1}z}$.

\textbf{Illustration of effects of displacement (figure~\ref{fig:exp(k+f)}).} The resulting plots are shown in figure~\ref{fig:exp(k+f)}. Each panel can be naturally divided by the line $a=c$ into an upper triangular region ($a<c$) and a lower triangular region ($a>c$). These two regions exhibit qualitatively distinct patterns, reflecting the stability properties of the underlying bosonic system. The eigenvalues of $K$ are given by $\pm\sqrt{a^2-c^2}$: in the upper triangular sector they are purely imaginary, corresponding to stable dynamics; whereas in the lower sector they are real, indicating instability.

The influence of the displacement term is most clearly revealed by comparing panels with different values of $\rho$ and $\tau$. Vertical comparisons between panels with different $\tau$ show that the direction of the displacement has only a minor effect on the overall structure. By contrast, horizontal comparisons across increasing $\rho$ demonstrate that the magnitude of the displacement controls the strength and extent of the induced fluctuations. These fluctuations originate predominantly in the unstable sector and begin to cross the boundary $a=c$ at around $\rho\sim1$. For large displacements, $\rho\gg1$, the effects become significant even in the stable region.

\subsection{Non-quadratic generators of Gaussian unitaries}

Gaussian unitaries were introduced via the relation~\eqref{eq:rep-inh} and reviewed as being generated by quadratic Hamiltonians. In particular, for a symplectic generator $K\in\mathfrak{sp}(2N,\mathbb{R})$ with corresponding group element $M=e^K$, the associated Gaussian unitaries $\mathcal{S}=\pm e^{\widehat{K}}$ implement the symplectic transformation $M$. However, it is also well-known that the exponential map from the Lie algebra $\mathfrak{sp}(2N,\mathbb{R})$ to the group $\mathrm{Sp}(2N,\mathbb{R})$ is not surjective, \ie there exist symplectic group elements $M$ that \emph{cannot} be written in the form $M=e^K$ for any $K\in\mathfrak{sp}(2N,\mathbb{R})$.

Despite this, every homogeneous unitary (\ie symplectic transformation) admits a unitary representation $\mathcal{S}(M)\in\mathcal{U}(\mathcal{H}_\infty)$ acting on the bosonic Hilbert space. This guarantees the existence of a (possibly non-quadratic) Hamiltonian $\hat{H}$ such that
\begin{align}\label{eq:S(M)}
    \mathcal{S}(M)=e^{-\ii\hat{H}}\,.
\end{align}
This naturally raises the question: what types of non-quadratic generators arise when the Gaussian unitary in~\eqref{eq:S(M)} with a quadratic Hamiltonians $\hat{H}$ alone?

We primarily concern the \emph{unreachable} region identified in figure~\ref{fig:fiber}, which is equivalent to
\begin{align}
    \tan^{-1}{\sinh{\rho}}<\tau<\pi-\tan^{-1}{\sinh{\rho}}\,,\label{eq:region-(III)}
\end{align}
as explained in~\cite[Eq.~(188)]{Hackl_2024}. In this unreachable region, the symplectic matrix (for a single bosonic mode) can be diagonalized to the form
\begin{align}
    M_\lambda=\begin{pmatrix}-\lambda&0\\0&-\frac{1}{\lambda}\end{pmatrix}=-\id_2\begin{pmatrix}\lambda&0\\0&\lambda^{-1}\end{pmatrix}=-\id_2e^K\,,
\end{align}
where $K=\log(\lambda)\,Y$. The corresponding unitary representation factorizes as
\begin{align}
    \mathcal{S}(M_\lambda)=\mathcal{S}(-\id_2)\,\mathcal{S}(e^K)=e^{\pm\ii\pi(\hat{n}+\frac{1}{2})}\,e^{\widehat{K}}
    \,.
\end{align}
If we now tried to combine the two generators into a single exponential, we cannot do so immediately, as $\hat{n}$ will generally not commute with $\widehat{K}$. This issue can be resolved by introducing the parity operator $\mathcal{P}=e^{\pm\ii\pi\hat{n}}$, which satisfies
\begin{subequations}
\begin{align}
    \mathcal{P}\ket{n}&=e^{\pm\ii\pi\hat{n}}\ket{n}=(-1)^n\ket{n}\,,
    \\
    e^{\ii\frac{\pi}{2}\mathcal{P}}\ket{n}&=e^{\ii\frac{\pi}{2}(-1)^n}\ket{n}=\ii\mathcal{P}\ket{n}\,,
\end{align}
\end{subequations}
The parity operator decomposes the Hilbert space into even and odd sectors of particle occupation, $\mathcal{H}=\mathcal{H}_\text{even}\oplus\mathcal{H}_\text{odd}$. It clearly commutes with a quadratic generator $\widehat{K}$, which follows from the fact that each term in $\widehat{K}$ either leaves the particle number invariant or changes it by $\pm 2$, thereby keeping parity invariant.

Using this operator, we may safely write
\begin{align}\begin{split}
    \mathcal{S}(-\id_2)&=e^{\pm\ii\pi(\hat{n}+\frac{1}{2})}=\pm\ii e^{\pm\ii\pi\hat{n}}\\
    &\overset{\text{choose }+}{=}\ii\mathcal{P}=e^{\ii\frac{\pi}{2}\mathcal{P}}\,.
\end{split}\end{align}
Since $\mathcal{P}$ commutes with any quadratic generator $\widehat{K}$, the unitary can be written compactly as
\begin{align}
    \mathcal{S}(M_\lambda)&=e^{\widehat{K}+\ii\frac{\pi}{2}\mathcal{P}}\,,
\end{align}
where the generator explicitly contains the non-quadratic parity operator.

It is natural to ask whether operators of the form $e^{\ii\tau\mathcal{P}}$, with $\tau\in\mathbb{R}$, define Gaussian unitaries in the sense of~\eqref{eq:rep-inh}. The answer is negative. Parity is a discrete Clifford operation, and $e^{\ii\tau\mathcal{P}}$ is Gaussian if and only if $\tau=k\frac{\pi}{2}$ for some integer $k$. For generic values of $\tau$, the operator $e^{\ii\tau\mathcal{P}}$ is non-Gaussian and corresponds to a parity-controlled non-quadratic evolution.

\begin{figure}[t]
    \centering
    \begin{tikzpicture}[scale=1.2]
        \draw[very thick,blue] plot[domain=0:6,samples=91,smooth] ({\x},{.01*(-atan(sinh(\x))+180)});
        \draw[very thick,blue] plot[domain=0:6,samples=91,smooth] ({\x},{.01*(atan(sinh(\x))-180)});
        \draw[very thick,red] plot[domain=0:6,samples=91,smooth] ({\x},{-.01*atan(sinh(\x))});
        \draw[very thick,red] plot[domain=0:6,samples=91,smooth] ({\x},{.01*atan(sinh(\x))});
        \fill[red,opacity=.2] plot[domain=0:6,samples=91,smooth] ({\x},{-.01*atan(sinh(\x))}) -- plot[domain=6:0,samples=91,smooth] ({\x},{.01*atan(sinh(\x))});
        \fill[blue,opacity=.2] plot[domain=0:6,samples=91,smooth] ({\x},{.01*(-atan(sinh(\x))+180)}) -- (6,1.8);
        \fill[blue,opacity=.2] plot[domain=0:6,samples=91,smooth] ({\x},{.01*(atan(sinh(\x))-180)}) -- (6,-1.8);
        \fill[green,opacity=.2] plot[domain=0:6,samples=91,smooth] ({\x},{.01*(-atan(sinh(\x))+180)}) -- plot[domain=6:0,samples=91,smooth] ({\x},{.01*atan(sinh(\x))});
        \fill[green,opacity=.2] plot[domain=0:6,samples=91,smooth] ({\x},{.01*(atan(sinh(\x))-180)}) -- plot[domain=6:0,samples=91,smooth] ({\x},{-.01*atan(sinh(\x))});
        \draw (.7,.9) node{(Future)} (.7,-.9) node{(Past)} (3.8,.4) node{(Reachable space)} (3.8,1.4) node{(Unreachable space)};        
        \draw[thick,->] (0,-1.8) node[left]{$-\pi$} -- (0,2);
        \draw (0,1.8) node[left]{$\pi$};
        \draw (0,2.1) node{$z$};
        \draw[thick,->] (0,0) node[left]{$0$} -- (6.2,0) node[below]{$\rho$};
        \draw (0,-1.8) rectangle (6,1.8);
        \draw[orange,thick,->] (0,0,0) -- (.5,0,0);
        \draw[orange,thick,->] (0,0,0) -- (0,.5,0) node[right]{$Z$};
        \draw[orange,thick,->] (0,0,0) -- (0,0,.8) node[below]{$Y$};
        \draw[orange] (0.6,0.05) node[below]{$X$};
        \draw[orange] (0.08,0) node[below]{$\id$};
    \end{tikzpicture}
    \ccaption{Planar section of bosonic fiber bundle for $\mathrm{Sp}(2,\mathbb{R})$}{A general symplectic group element $M$ is parametrized by $(\rho,\theta,z)$, where $\rho$ and $\theta$ are polar coordinates of a horizontal plane and $z$ is a $2\pi$-periodic vertical coordinate, as $M=
    \cosh{\rho}\left(\begin{smallmatrix}\cos{z}&\sin{z}\\-\sin{z}&\cos{z}\end{smallmatrix}\right)-\sinh{\rho}\left(\begin{smallmatrix}-\sin{\theta}&\cos{\theta}\\\cos{\theta}&\sin{\theta}\end{smallmatrix}\right)$. We demonstrate a section of $\theta=0$. We also recognize the Lie algebra $\mathfrak{sp}(2,\mathbb{R})=\mathrm{span}(X,Y,Z)$ as tangent space at the identity. We further illustrate the different regions.}
    \label{fig:fiber}
\end{figure}

\section{Conclusion}\label{sec:conc}

We resolve the ambiguity of the global phase in composite Gaussian circuits by constructing a unified phase formalism for inhomogeneous Gaussian unitaries. Building on earlier parametrizations based on the circle function $\varphi(M)$ and the cocycle $\eta(M_1,M_2)$, we incorporate the displaced-squeezing phase $\gamma(M,z)$ to obtain the generalized inhomogeneous cocycle function $\zeta(M_1,M_2,z_1,z_2)$, which treats symplectic transformations and displacements consistently. These ingredients combine into a single phase function $\Psi(M,z)$ that lifts the affine symplectic parameters $(M,z)$ to the extended element $(M,z,\Psi)$ and yields a unique and associative group product using $\zeta$ defined in~\eqref{eq:def:zeta}. Our framework determines the complex phase of arbitrary Gaussian unitaries and their compositions, enabling phase-correct operator multiplication and measurable interference predictions. Previous works have addressed special cases of Gaussian phases~\cite{PhysRevA.61.022306,10.21468/SciPostPhys.17.3.082}, whereas our approach provides a general and systematic solution for the fully inhomogeneous setting focusing on the multiplication laws.

As a second main result, we derive explicit phase formulas for unitaries generated by general quadratic Hamiltonians. Rather than decomposing transformations into elementary gates, the global phase of $e^{-\ii\hat{H}}$ is obtained directly from its generator, providing compact expressions suited for numerical implementation. This complements recent methods for evaluating Gaussian overlaps~\cite{10.21468/SciPostPhys.17.3.082} while supplying the phase information required for exact circuit composition. As a direct application, phase-coherent simulation of optical circuits becomes possible even for long sequences of displaced-squeezing operations. Although we focus on bosonic systems, the same derivations extend to fermionic Gaussian unitaries with minor modifications; in particular, the generator-based phase formula provides the key ingredient needed to lift affine orthogonal transformations to the Pin group.

Beyond the technical contributions, our work connects mathematical and practical treatments of Gaussian quantum mechanics. Physics-oriented approaches often address specific circuits where phases are fixed case by case, while rigorous treatments of the metaplectic and Weyl groups~\cite{woit2017quantum,folland1989harmonic,Simon1993,zhang1990coherent,derezinski2013mathematics} are typically less accessible to practitioners. We translate this structure into explicit formulas that can be directly implemented in modeling and simulation. This offers a bridge between abstract representation theory and experimental needs, complementing phase-aware Gaussian computation frameworks developed by~\cite{Dias:2023arXiv230712912D,Dias:2024arXiv240319059D,PhysRevA.61.022306,bravyi_complexity_2017,10.21468/SciPostPhys.17.3.082,Quesada2025}.

Our formalism is particularly relevant for non-Gaussian variational methods that rely crucially on phase-sensitive quantities. Superpositions of Gaussian states have been studied extensively as efficient variational representations of many-body states~\cite{bravyi_complexity_2017} and as tools for simulation and tomography of quantum systems~\cite{Dias:2023arXiv230712912D,Dias:2024arXiv240319059D}. In these approaches, overlaps between different Gaussian branches determine norms, energies, and gradients, so a well-defined relative phase between Gaussian unitaries is essential. Similarly, generalized Gaussian states introduced in~\cite{shi2018variational} and further developed in~\cite{hackl_geometry_2020,guaita_generalization_2021} require the evaluation of expectation values and geometric quantities that depend explicitly on the phases of underlying Gaussian transformations. Without a consistent global-phase convention, interference between Gaussian components becomes ambiguous and variational optimization ill-defined. The inhomogeneous lifting $(M,z,\Psi)$ developed here provides precisely this missing structure, enabling unambiguous computation of generalized Wick expectation values, overlaps, and tangent-space objects required for gradient-based optimization on enlarged variational manifolds.

A further application is efficient simulation of $n$-mode Gaussian circuits with correct global phases. Many numerical tools ignore phases or fix them heuristically, leading to inconsistencies when circuits are concatenated or optimized. The closed-form phase function $\Psi$ and its composition rule via $\zeta$ enable reproducible, phase-aware circuit simulation with improved stability. This directly benefits tasks such as circuit compilation and interferometric phase estimation.

In summary, this work shows that phase information---often treated as a technical nuisance---can be made explicit, computable, and structurally transparent for Gaussian quantum mechanics. We expect these results to facilitate reliable phase-aware simulation, variational optimization, and the integration of Gaussian methods into larger non-Gaussian quantum algorithms.

\acknowledgments

JS acknowledges the Melbourne Research Scholarship provided by the School of Mathematics and Statistics and the Faculty of Science at the University of Melbourne. LH thanks Tommaso Guaita and Thomas Quella for helpful discussions. The research of LH is supported by ‘The Quantum Information Structure of Spacetime’ Project (QISS), by grant $\#$62312 from the John Templeton Foundation.

\appendix

\section{Background and notation}

This section of the appendix is to introduce our notations and conventions, as well as the theorems we used.

\subsection{Operators and representations}\label{ap:A:operator}

In Section~\ref{sec:intro} and Section~\ref{sec:2:structures}, we present our quadrature operators $\hat{\xi}^a$ in terms of position $\hat{q}$ and momentum $\hat{p}$. In fact, there is another standard basis that is regularly expressed in quantum physics, the ladder operators. Here we present them parallel, especially for~\eqref{eq:xi} and~\eqref{eq:def:matrix}.

We introduce annihilation and creation operators $\hat{a}_i^\dagger$ and $\hat{a}_j$ satisfying the commutation relation $[\hat{a}_i,\hat{a}^\dagger_j]=\delta_{ij}$ for bosons. We can construct the Hermitian operators
\begin{subequations}\begin{align}
    \hat{q}_j={\textstyle\frac{1}{\sqrt{2}}}(\hat{a}^\dagger_j+\hat{a}_j)\,,\qquad\hat{p}_j={\textstyle\frac{\ii}{\sqrt{2}}}(\hat{a}^\dagger_j-\hat{a}_j)\,,\label{eq:qp_basis}\\
    \hat{a}_j={\textstyle\frac{1}{\sqrt{2}}}(\hat{q}_j+\ii\hat{p}_j)\,,\qquad\hat{a}^\dagger_j={\textstyle\frac{1}{\sqrt{2}}}(\hat{q}_j-\ii\hat{p}_j)\,,\label{eq:ad_basis}
\end{align}\end{subequations}
We refer to them as real basis~\eqref{eq:qp_basis} in $(\hat{q},\hat{p})$ and complex basis~\eqref{eq:ad_basis} in $(\hat{a},\hat{a}^\dagger)$. Note that the main body of this manuscript only uses the real basis. We collect these operators into a single operator-valued phase space vector $\hat{\xi}^a$, given in the two standard bases by
\begin{align}
    \hat{\xi}^a\equiv\begin{cases}
        (\hat{q}_1,\cdots,\hat{q}_N,\hat{p}_1,\cdots,\hat{p}_N)\,,&\textbf{(real basis)}\\
        (\hat{a}_1,\cdots,\hat{a}_N,\hat{a}^\dagger_1,\cdots,\hat{a}^\dagger_N)\,,&\textbf{(complex basis)}
    \end{cases}
\end{align}
where we indicate the respective real and complex basis of the expression. With the above bases in terms of $\hat{\xi}^a$, the matrix representations of structures can be summarized in table~\ref{tab:matrices}.

\begin{table}[t]
    \centering
    \renewcommand{\arraystretch}{1.5}
\begin{tabular}{c *2{>{\renewcommand{\arraystretch}{1}}c}} 
\hline
\hline
 basis of $\hat{\xi}^a$ & \quad Real $(\hat{q},\hat{p})$ & \quad Complex $(\hat{a},\hat{a}^\dagger)$ \\
 \hline
 & & \\[-1em]
 \makecell{symplectic\\ form $\Omega^{ab}$} & $\begin{pmatrix}0&\id\\-\id&0\end{pmatrix}$ & $\ii\begin{pmatrix}0&-\id\\\id&0\end{pmatrix}$ \\[-1em]
 & & \\
 $\omega_{ab}=(\Omega^{-1})_{ab}$ & $\begin{pmatrix}0&-\id\\\id&0\end{pmatrix}$ & $\ii\begin{pmatrix}0&-\id\\\id&0\end{pmatrix}$ \\[-1em]
 & & \\
 \hline
 & & \\[-1em]
 \makecell{positive-definite\\ metric $G^{ab}$} & $\begin{pmatrix}\id&0\\0&\id\end{pmatrix}$ & $\begin{pmatrix}0&\id\\\id&0\end{pmatrix}$ \\[-1em]
 & & \\
 $g_{ab}=(G^{-1})_{ab}$ & $\begin{pmatrix}\id&0\\0&\id\end{pmatrix}$ & $\begin{pmatrix}0&\id\\\id&0\end{pmatrix}$ \\[-1em]
 & & \\
 \hline
 & & \\[-1em]
 \makecell{complex\\structure $J^a{}_b$} & $\begin{pmatrix}0&\id\\-\id&0\end{pmatrix}$ & $\ii\begin{pmatrix}-\id&0\\0&\id\end{pmatrix}$ \\[-1em]
 & & \\
\hline
\hline
\end{tabular}
\caption{The matrix representations and respective inverses of the sympletic form $\Omega$, the positive-definite metric $G$ and the complex structure $J$ with respect to real and complex bases of $\hat{\xi}^a$.}
    \label{tab:matrices}
\renewcommand{\arraystretch}{1}
\end{table}

\subsection{Powers of adjoint}\label{ap:A:adjoint}
In the proofs of the following sections, we will encounter power series of nested commutators (appendix~\ref{ap:A:BCH} and~\ref{ap:B:nested}), which are best understood as powers of the adjoint action. We consider a Lie algebra $\mathfrak{g}$ with its Lie bracket $[\cdot,\cdot]$. We can define the adjoint representation $\adrep{X}$ of $X\in\mathfrak{g}$ as $\adrep{X}Y=[X,Y]$. An alternative notation is
\begin{subequations}
\begin{align}
    [[X]^n,Y]&=[\underbrace{X,\dots[X,[X}_{n\text{ times of }X},Y]]\cdots]=(\adrep{X})^nY\,,\label{adn.left}
    \\
    [Y,[X]^n]&=[\cdots[[Y,\underbrace{X],X]\dots,X}_{n\text{ times of }X}]=(\adrep{-X})^nY\,.\label{adn.right}
\end{align}
\end{subequations}
Note that (\ref{adn.left}) has the regular commutator sequence from right to left, while (\ref{adn.right}) is from left to right. They are related by
\begin{align}
    [Y,[X]^n]&=[[-X]^n,Y]=(-1)^n[[X]^n,Y]\,.
\end{align}

\subsection{Baker-Campbell-Hausdorff (BCH) formula}\label{ap:A:BCH}

Here we review several standard formulas that we will use in the proof of Lemma~\ref{lemma:BCH:exp(w)exp(K)} in appendix~\ref{ap:ab}.

\textbf{Dynkin's formula.}
The explicit Campbell-Hausdorff formula in Dynkin' form~\cite{Serre} is given by the following general combinatorial formula
\begin{align}\begin{split}\label{eq:Dynkin}
    &\log(e^Xe^Y)=\sum_{n=1}^\infty\frac{(-1)^{n-1}}{n}\times\\&\qquad\sum_{\substack{r_1+s_1>0\\ \vdots\\r_n+s_n>0}}\frac{[[X]^{r_1},[[Y]^{s_1},\cdots[[X]^{r_n},[Y]^{s_n}]\cdots]]}{\big(\sum_{j=1}^n(r_j+s_j)\big)\cdot\prod_{i=1}^nr_i!s_i!}\,.
\end{split}\end{align}
With the understanding that $[X]=X$. Since $[X,X]=0$, the term will also be zero if $s_n>1$ or if $s_n=0$ and $r_n>1$. In addition we can also define $r=\sum_{j=1}^n r_j$ and $s=\sum_{j=1}^n s_j$ schematically.

\textbf{Campbell identity.} The Campbell identity is $\ADrep{e^X}=e^{\adrep{X}}$, where $\ADrep{A}Y=AYA^{-1}$. With this identity, the Proposition 3.35 of~\cite{Hall2015} presents a special case of the BCH formula
\begin{align}
    e^XYe^{-X}=\ADrep{e^X}Y=e^{\adrep{X}}Y
    =\sum_{n=0}^\infty\frac{[[X]^n,Y]}{n!}\,.
\end{align}

\textbf{Differential equations of Lie algebra.} We will need the derivative of certain exponentials of Lie algebra elements. A special form is evaluated by
\begin{align}\begin{split}
    &\frac{\id-e^{-\adrep{X}}}{\adrep{X}}(Y)=\frac{e^{\adrep{-X}}-\id}{\adrep{-X}}(Y)\\
    &=\sum_{n=0}^\infty\frac{(\adrep{-X})^n}{(n+1)!}(Y)=\sum_{n=0}^\infty\frac{[Y,[X]^n]}{(n+1)!}\,.
\end{split}\end{align}
If $X(t)$ is a smooth Lie algebra-valued function with derivative $\dot{X}(t)=\frac{\mathrm{d}X(t)}{\mathrm{d}t}$, by Theorem 5.4 of~\cite{Hall2015}, its exponential $e^{X(t)}$ satisfies the differential equation
\begin{align}
    e^{-X(t)}\frac{\mathrm{d}}{\mathrm{d}t}(e^{X(t)})=\frac{\id-e^{-\adrep{X(t)}}}{\adrep{X(t)}}\big(\dot{X}(t)\big)=\sum_{n=0}^\infty\frac{[\dot{X},[X]^n]}{(n+1)!}\,.\label{eq:adxy}
\end{align}

\section{Lie Algebra and Commutators}\label{ap:commutators}

In this appendix, we provide detailed calculations for all Lie algebraic commutation results of the family of linear and quadratic operators. These are essential for fulfilling intermediate steps in many derivation in the main text and other appendix. Recall our boson-fermion unified notation, the operators are denoted by
\begin{align}
    \widehat{z}&=-z^a\lambda_{ab}\hat{\xi}^b\,, &&\textbf{(linear vector)}\\
    \widehat{w}&=-\iii w_a\hat{\xi}^a\,, &&\textbf{(linear co-vector)}\\
    \widehat{K}&=\tfrac{1}{2}\lambda_{ac}K^c{}_b\hat{\xi}^a\hat{\xi}^b\,. &&\textbf{(quadratic)}
\end{align}

\subsection{Commutators between linear operators}\label{ap:B:lin-com}

This sub-appendix provides technical details particularly for Section~\ref{sec:2:disp}. We have
\begin{align}\begin{split}\label{eq:ap:xi_com_z}
    [\hat{\xi}^b,\widehat{z}]&=[\hat{\xi}^b,-z^c\lambda_{cd}\hat{\xi}^d]=[\hat{\xi}^b,\hat{\xi}^d\lambda_{dc}z^c]
    \\&=\bfcase{[\hat{\xi}^b,\hat{\xi}^d]\lambda_{dc}z^c}{\{\hat{\xi}^b,\hat{\xi}^d\}\lambda_{dc}z^c}
    \\&=\Lambda^{bd}\lambda_{dc}z^c=z^b\,.
\end{split}\end{align}
The commutator between two linear operators will give a scalar function as
\begin{align}\begin{split}\label{eq:ap:z_com}
    [\widehat{z}_1,\widehat{z}_2]&=[-z_1^a\lambda_{ab}\hat{\xi}^b,\widehat{z}_2]=-z_1^a\lambda_{ab}[\hat{\xi}^b,\widehat{z}_2]
    \\&=-z_1^a\lambda_{ab}z_2^b=-z_1\lambda z_2=z_2\lambda z_1\,.
\end{split}\end{align}
This means that any higher-order commutator of linear operators would vanish. Let's define $\tilde{w}^a\equiv[\hat{\xi}^a,\widehat{w}]$, then
\begin{align}\begin{split}
    \iii\tilde{w}^a&=\iii[\hat{\xi}^a,\widehat{w}]=[\hat{\xi}^a,-\iii^2w_b\hat{\xi}^b]=[\hat{\xi}^a,\hat{\xi}^bw_b]
    \\&=\bfcase{[\hat{\xi}^a,\hat{\xi}^b]w_b}{\{\hat{\xi}^a,\hat{\xi}^b\}w_b}
    \\&=\Lambda^{ab}w_b\,.
\end{split}\end{align}
It follows that
\begin{align}\begin{split}\label{eq:ap:w-com}
    [\widehat{w}_1,\widehat{w}_2]&=[-\iii w_{1;a}\hat{\xi}^a,\widehat{w}_2]=-w_{1;a}\iii[\hat{\xi}^a,\widehat{w}_2]
    \\&=-w_{1;a}\Lambda^{ab}w_{2;b}=-w_1\Lambda w_2=w_2\Lambda w_1\,.
\end{split}\end{align}
It follows from the famous BCH formula, we get
\begin{subequations}\label{eq:ap:DD}
\begin{align}
    \begin{split}
    \log\big(e^{\widehat{z}_1}e^{\widehat{z}_2}\big)&=\widehat{z}_1+\widehat{z}_2+\textstyle{\frac{1}{2}}[\widehat{z}_1,\widehat{z}_2]\\&=\widehat{z}_1+\widehat{z}_2+\textstyle{\frac{1}{2}}z_2\lambda z_1\,,
    \end{split}
    \\
    \begin{split}
    \log\big(e^{\widehat{w}_1}e^{\widehat{w}_2}\big)&=\widehat{w}_1+\widehat{w}_2+\textstyle{\frac{1}{2}}[\widehat{w}_1,\widehat{w}_2]\\&=\widehat{w}_1+\widehat{w}_2+\textstyle{\frac{1}{2}}w_2\Lambda w_1\,.
    \end{split}
\end{align}
\end{subequations}
The other form of BCH tell us
\begin{subequations}\label{eq:ap:disp}
\begin{align}
    e^{-\widehat{z}}\hat{\xi}^ae^{\widehat{z}}&=\hat{\xi}^a+[\hat{\xi}^a,\widehat{z}]=\hat{\xi}^a+z^a\,,
    \\
    e^{-\widehat{w}}\hat{\xi}^ae^{\widehat{w}}&=\hat{\xi}^a+[\hat{\xi}^a,\widehat{w}]=\hat{\xi}^a+\tilde{w}^a\,,
\end{align}
\end{subequations}
hence the equivalent displacement vector is $\tilde{w}^a$.

\subsection{Commutators with quadratic operators}\label{ap:B:qua-com}

This sub-appendix is the foundation to appendix~\ref{ap:B:nested} as well as basic properties of squeezing operators. Using the fact that
\begin{align}
    \lambda K\Lambda=\left.\bfcase{\omega K\Omega}{gKG}\right\}=-K^\tp\,,
\end{align}
expanding indices of $\lambda K\Lambda=-K^\tp$ tells us
\begin{align}
    \lambda_{ab}K^b{}_c\Lambda^{cd}&=(\lambda K\Lambda)_a{}^d=(-K^\tp)_a{}^d=-K^d{}_a\,.
\end{align}
For a polynomial function (power series) of $K$, say $f(K)$, by inserting pairs such as $K^2=K\Lambda\lambda K$, we would have
\begin{align}
\lambda_{ab}f(K)^b{}_c\Lambda^{cd}=f(-K)^d{}_a\,.\label{eq:okof}
\end{align}

For the quadratic operators, we firstly compute
\begin{align}\begin{split}
    [\hat{\xi}^a,\hat{\xi}^b\hat{\xi}^c]&=\bfcase{[\hat{\xi}^a,\hat{\xi}^b]\hat{\xi}^c+\hat{\xi}^b[\hat{\xi}^a,\hat{\xi}^c]}{\{\hat{\xi}^a,\hat{\xi}^b\}\hat{\xi}^c-\hat{\xi}^b\{\hat{\xi}^a,\hat{\xi}^c\}}
    \\&=\bfcase{[\hat{\xi}^a,\hat{\xi}^b]\hat{\xi}^c-\hat{\xi}^b[\hat{\xi}^c,\hat{\xi}^a]}{\{\hat{\xi}^a,\hat{\xi}^b\}\hat{\xi}^c-\hat{\xi}^b\{\hat{\xi}^c,\hat{\xi}^a\}}
    \\&=\Lambda^{ab}\hat{\xi}^c-\Lambda^{ca}\hat{\xi}^b\,.
\end{split}\end{align}
This means the commutator of a linear operator with a quadratic operator will give a linear operator. For both bosons and fermions we have
\begin{align}\begin{split}
    [\hat{\xi}^a,\widehat{K}]&=\textstyle{\frac{1}{2}}\lambda_{bd}K^d{}_c[\hat{\xi}^a,\hat{\xi}^b\hat{\xi}^c]\\&=\textstyle{\frac{1}{2}}\lambda_{bd}K^d{}_c(\Lambda^{ab}\hat{\xi}^c-\Lambda^{ca}\hat{\xi}^b)\\&=\textstyle{\frac{1}{2}}(\Lambda^{ab}\lambda_{bd}K^d{}_c\hat{\xi}^c- \lambda_{bd}K^d{}_c\Lambda^{ca}\hat{\xi}^b)\\&=\textstyle{\frac{1}{2}}(\delta^a{}_d K^d{}_b+K^a{}_b)\hat{\xi}^b\\&=K^a{}_b\hat{\xi}^b\,.
\end{split}\end{align}
Thus, we know
\begin{align}
    [\widehat{w},\widehat{K}]=-\iii w_a[\hat{\xi}^a,\widehat{K}]=-\iii w_aK^a{}_b\hat{\xi}^b=\widehat{wK}\,,
\end{align}
where $(wK)_a=w_bK^b{}_a$. 

We can also know the commutator for two quadratic operators
\begin{align}\begin{split}
    [\widehat{K}_1,\widehat{K}_2]&=\textstyle{\frac{1}{2}}\lambda_{ac}K^c_{1;b}[\hat{\xi}^a\hat{\xi}^b,\widehat{K}_2]
    \\&=\textstyle{\frac{1}{2}}\lambda_{ac}K^c_{1;b}(\hat{\xi}^a[\hat{\xi}^b,\widehat{K}_2]+[\hat{\xi}^a,\widehat{K}_2]\hat{\xi}^b)
    \\&=\textstyle{\frac{1}{2}}\lambda_{ac}K^c_{1;b}(K^b_{2;d}\hat{\xi}^a\hat{\xi}^d+K^a_{2;d}\hat{\xi}^d\hat{\xi}^b)
    \\&=\textstyle{\frac{1}{2}}(\lambda_{ac}(K_1K_2)^c{}_b+\lambda_{dc}K^d_{2;a}K^c_{1;b})\hat{\xi}^a\hat{\xi}^b
    \\&=\textstyle{\frac{1}{2}}\lambda_{ac}(K_1K_2-K_2K_1)^c{}_b\hat{\xi}^a\hat{\xi}^b
    \\&=\widehat{[K_1,K_2]}\,,
\end{split}\end{align}
where in the forth to fifth line, we use
\begin{align}\begin{split}
    \lambda_{dc}K^d_{2;a}K^c_{1;b}&=\lambda_{de}K^d_{2;f}\Lambda^{fc}\lambda_{ca}K^e_{1;b}\\&=\lambda_{ed}K^d_{2;f}\Lambda^{fc}\lambda_{ac}K^e_{1;b}\\&=\lambda_{ac}(-K_2)^c{}_eK^e_{1;b}\\&=\lambda_{ac}(-K_2K_1)^c{}_b\,.
\end{split}\end{align}
This result justifies the fact that $\widehat{K}$ forms a Lie algebra representation. 

\subsection{Nested commutators}\label{ap:B:nested}

Due to the exponential relations between Lie groups and Lie algebra, in this section, we analyze the series of algebraic elements, which emerged from nested commutators. These are particularly crucial for the proof of Lemma~\ref{lemma:BCH:exp(w)exp(K)}, which displayed in appendix~\eqref{ap:ab}. A simple induction start from
\begin{align}\begin{split}
    [\hat{\xi}^a,[\widehat{K}]^2]&=[[\hat{\xi}^a,\widehat{K}],\widehat{K}]=K^a{}_c[\hat{\xi}^c,\widehat{K}]\\&=K^a{}_cK^c{}_b\hat{\xi}^b=(K^2)^a{}_b\hat{\xi}^b\,,
\end{split}\end{align}
so that it implies that
\begin{align}
    [\hat{\xi}^a,[\widehat{K}]^n]=(K^n)^a{}_b\hat{\xi}^b\,.
\end{align}
This time, the higher order terms survive in the BCH formula, given by
\begin{align}\begin{split}
    e^{-\widehat{K}}\hat{\xi}^ae^{\widehat{K}}&=\sum_{n=0}^\infty\frac{[[-\widehat{K}]^n,\hat{\xi}^a]}{n!}=\sum_{n=0}^\infty\frac{[\hat{\xi}^a,[\widehat{K}]^n]}{n!}\\&=\sum_{n=0}^\infty\frac{(K^n)^a{}_b}{n!}\hat{\xi}^b=(e^{K})^a{}_b\hat{\xi}^b\,.
\end{split}\end{align}
Thus we have
\begin{align}\begin{split}
    e^{-\widehat{K}}\widehat{w}e^{\widehat{K}}&=-\iii w_ae^{-\widehat{K}}\hat{\xi}^ae^{\widehat{K}}=-\iii w_a(e^{K})^a{}_b\hat{\xi}^b\\&=-\iii(we^{K})_b\hat{\xi}^b=\widehat{we^K}\,.
\end{split}\end{align}
This also means
\begin{align}
    e^{\widehat{w}}e^{\widehat{K}}=e^{\widehat{K}}e^{-\widehat{K}}e^{\widehat{w}}e^{\widehat{K}}=e^{\widehat{K}}e^{\widehat{we^K}}\,.
\end{align}
Together, we have
\begin{align}\begin{split}
    (e^{\widehat{w}}e^{\widehat{K}})^{-1}\hat{\xi}^a(e^{\widehat{w}}e^{\widehat{K}})
    &=e^{-\widehat{K}}e^{-\widehat{w}}\hat{\xi}^ae^{\widehat{w}}e^{\widehat{K}}
    \\
    &=e^{-\widehat{K}}\big(\hat{\xi}^a+\tilde{w}^a\big)e^{\widehat{K}}
    \\
    &=(e^{K})^a{}_b\hat{\xi}^b+\tilde{w}^a\,,
\end{split}\end{align}
and
\begin{align}\begin{split}
    (e^{\widehat{K}}e^{\widehat{w}})^{-1}\hat{\xi}^a(e^{\widehat{K}}e^{\widehat{w}})
    &=e^{-\widehat{w}}e^{-\widehat{K}}\hat{\xi}^ae^{\widehat{K}}e^{\widehat{w}}
    \\
    &=e^{-\widehat{w}}\big((e^{K})^a{}_b\hat{\xi}^b\big)e^{\widehat{w}}
    \\
    &=(e^{K})^a{}_b\hat{\xi}^b+(e^{K})^a{}_b\tilde{w}^b\,.
\end{split}\end{align}
The above arguments justify that $e^{\widehat{w}}e^{\widehat{K}}$ has better brevity. That is, we would prefer our unitary operators to apply the squeezing and rotation operator first, followed by a displacement operator. Thus, as part of the proof of the decomposition, we have
\begin{align}\begin{split}\label{eq:proof:uxiu}
    &\mathcal{U}^\dagger(M,z,\Psi)\hat{\xi}^a\mathcal{U}(M,z,\Psi)
    \\
    &=e^{-\ii\theta}e^{\ii\theta}\mathcal{S}^\dagger(M,\psi)\mathcal{D}^\dagger(z)\hat{\xi}^a\mathcal{D}(z)\mathcal{S}(M,\psi)
    \\
    &=\mathcal{S}^\dagger(M,\psi)(\hat{\xi}^a+z^a)\mathcal{S}(M,\psi)
    \\
    &=M^a{}_b\hat{\xi}^b+z^a\,,
\end{split}\end{align}
where we use~\eqref{eq:def:disp} for the third line and~\eqref{eq:def:sq-psi-op} for the last line. 

In the other direction, we have
\begin{subequations}
\begin{align}
    [\widehat{w},[\widehat{K}]^n]&=-\iii w_a[\hat{\xi}^a,[\widehat{K}]^n]=\widehat{wK^n}\,,
    \\
    [[\widehat{K}]^n,\widehat{w}]&=[\widehat{w},[-\widehat{K}]^n]=\widehat{w(-K)^n}\,.
\end{align}
\end{subequations}
Therefore, 
\begin{align}
    [\widehat{w},[\widehat{w},[\widehat{K}]^n]]=[\widehat{w},\widehat{wK^n}]=wK^n\Lambda w\,.
\end{align}

\begin{lemma}
    $[\widehat{w},[\widehat{w},[\widehat{K}]^n]]=0$ for all even $n$. 
\end{lemma}
\begin{proof}
It has
\begin{align}\begin{split}
    wK^n\Lambda w&=w\Lambda\lambda K^n\Lambda w\\&=w\Lambda(-K^\tp)^nw\\
    &=(-1)^nw\Lambda(wK^n)\\wK^n\Lambda w&=-(-1)^nwK^n\Lambda w\\ (1+(-1)^n)wK^n\Lambda w&=0\,,
\end{split}\end{align}
which means $wK^n\Lambda w=0$ for even $n$.
\end{proof}

\subsection{Commutators with Lie group elements}\label{ap:B:Lie-com}

After Lie group elements formed from Lie algebra through exponentials, we need to have a full dictionary of the commutation laws between operators and Lie group elements. For a Lie group element $e^{K}=M\in\mathcal{G}$, we have
\begin{equation}
    \begin{aligned}
        \lambda &=M^\tp\lambda M\,,
        \\
        \lambda M&=M^{-\tp}\lambda\,,
        \\
        0&=\lambda K+K^\tp\lambda\,,
    \end{aligned}
    \qquad
    \begin{aligned}
        \Lambda&=M\Lambda M^\tp\,,
        \\
        M\Lambda&=\Lambda M^{-\tp}\,,
        \\
        0&=K\Lambda+\Lambda K^\tp\,.
    \end{aligned}
\end{equation}
Expanding in indices as $M^a{}_b=(e^{K})^a{}_b$, they are
\begin{subequations}
\begin{align}
    \lambda_{ab}&=(M^\tp)_a{}^c\lambda_{cd}M^d{}_b=M^c{}_a\lambda_{cd}M^d{}_b\,,
    \\
    \Lambda^{ab}&=M^a{}_c\Lambda^{cd}(M^\tp)_d{}^b=M^a{}_c\Lambda^{cd}M^b{}_d\,;
    \\
    \lambda_{ac}M^c{}_b&=(M^{-\tp}\lambda)_{ab}=(M^{-1})^c{}_a\lambda_{cb}\,,
    \\
    M^a{}_c\Lambda^{cb}&=(\Lambda M^{-\tp})^{ab}=\Lambda^{ac}(M^{-1})^b{}_c\,.
\end{align}
\end{subequations}
A direct result is
\begin{align}\begin{split}
    [\widehat{Mz_1},\widehat{z_2}]&=z_2\lambda Mz_1=z_2M^{-\tp}\lambda z_1
    \\&=M^{-1}z_2\lambda z_1=[\widehat{z_1},\widehat{M^{-1}z_2}]\,,
\end{split}\end{align}
and
\begin{align}\begin{split}
    [\widehat{w_1M},\widehat{w_2}]&=-w_1M\Lambda w_2=-w_1\Lambda M^{-\tp} w_2
    \\&=-w_1\Lambda w_2M^{-1}=[\widehat{w_1},\widehat{w_2M^{-1}}]\,.
\end{split}\end{align}
This means
\begin{align}
    [\widehat{w_1M},\widehat{w_2M}]=[\widehat{w}_1,\widehat{w}_2]=[\widehat{w_1e^K},\widehat{w_2e^K}]\,.
\end{align}
Now consider a function $F^a{}_b$ of $K^a{}_b$ or $M^a{}_b$. If $w_a=\tilde{z}_a=\iii\lambda_{ab}z^b$ and so $\iii z^a=\Lambda^{ab}w_b$, we have
\begin{align}\begin{split}\label{eq:ap:F-com}
    [\widehat{w},\widehat{wF}]&=wF\Lambda w=w_aF^a{}_c\Lambda^{cb}w_b
    \\
    &=\iii\lambda_{ab}z^bF^a{}_c\iii z^c
    \\
    &=z^b\lambda_{ba}F^a{}_cz^c=z\lambda Fz
    \\
    &=[\widehat{Fz},\widehat{z}]\,.
\end{split}\end{align}

\section{Proof of Lemma~\ref{lemma:gamma}}\label{ap:gamma}

Our goal is to prove the following expectation value for bosons
\begin{align}\begin{split}
    &\braket{J|\mathcal{D}(z)\mathcal{S}(T,1)|J}=\braket{0|e^{\widehat{z}}e^{\widehat{K}_+}|0}\\&=\det(\id-L^2)^{1/8}e^{-\frac{1}{4}zg(\id-L)z}e^{-\frac{\ii}{4}z\omega Lz}\,,
\end{split}\end{align}
where we treat the Gaussian state as a vacuum $\ket{0}=\ket{J,z}$ and $L=-Y_M$ is a linear map. 

\begin{proof}
This proof can be constructed by primarily adopting results and formulas in~\cite[Section III-A-3]{Hackl_2021}. The proof will also be performed for both bosons and fermions. We shall state several ingredients to support a smooth proof.

\textbf{The linear map.} A key linear map is introduced as~\cite[Eq. (157)]{Hackl_2021}, which has multiple variations
\begin{align}\begin{split}
    L&=\tanh(K_+)=\tanh({\textstyle\frac{1}{2}}\log\Delta_M)
    \\
    &=\frac{\Delta_M-\id}{\Delta_M+\id}=\id-\frac{2}{\id+\Delta_M}
    \\
    -L&=\frac{\id-\Delta_M}{\id+\Delta_M}=Y_{M}=Z_{M^{-1}}
    \,,
\end{split}\end{align}
where we use the fact
\begin{align}
    \tanh x=\frac{e^x-e^{-x}}{e^x+e^{-x}}=\frac{e^{2x}-1}{e^{2x}+1}\,.
\end{align}

\textbf{The switch tilde notation.} As a contrast to our unified notation, the underlying symmetry of bosons and fermions would also urge us to form a switch notation for convenience. These are summarized in table~\ref{tab:switch}, as an extension to table~\ref{tab:structures}. It is easy to notice from~\eqref{eq:Kaehler-relation} that $\widetilde{\Lambda}^{cd}\lambda_{db}=\ii J^c{}_b$ and $\ii \lambda_{ac}J^c{}_b=\widetilde{\lambda}_{ab}$, hence $\lambda_{ac}\widetilde{\Lambda}^{cd}\lambda_{db}=\widetilde{\lambda}_{ab}$. These observation is essential for the next paragraph.

\begin{table}[t]
    \centering
    \renewcommand{\arraystretch}{1.5}
\begin{tabular}{c|c c c} 
\hline
\hline
 \textbf{notation} & \textbf{structures} & \textbf{bosons} & \textbf{fermions} \\
\hline
    \multirow{2}{*}{\centering\textbf{unified}} & $\Lambda^{ab}$ & $\ii\Omega^{ab}$ & $G^{ab}$ \\
    & $\lambda_{ab}=(\Lambda^{-1})_{ab}$ & $-\ii\omega_{ab}$ & $g_{ab}$ \\
    \hline
    \multirow{2}{*}{\centering\textbf{switch}}& $\widetilde{\Lambda}^{ab}$ & $G^{ab}$ & $\ii\Omega^{ab}$ \\
    & $\widetilde{\lambda}_{ab}=(\widetilde{\Lambda}^{-1})_{ab}$ & $g_{ab}$ & $-\ii\omega_{ab}$ \\
\hline
\hline
\end{tabular}
\ccaption{The switch tilde notation}{A wide tilde notation means a switch between the structures for bosons and fermions. Salute to supersymmetry.}
\label{tab:switch}
\renewcommand{\arraystretch}{1}
\end{table}

\textbf{The 2-point function.} The two point function is also required from~\cite[Eq. (109)]{Hackl_2021}
\begin{subequations}
\begin{align}
    C_2^{ab}&=\tfrac{1}{2}(G^{ab}+\ii\Omega^{ab})\\
    (C_2^\tp)^{ab}&=\tfrac{1}{2}(G^{ab}-\ii\Omega^{ab})\,.
\end{align}
\end{subequations}
By using the switch notation, we can quickly notice the following facts
\begin{subequations}
\begin{align}
    2C_2^{cd}\lambda_{db}&=(\Lambda^{cd}+\widetilde{\Lambda}^{cd})\lambda_{db}=\delta^c{}_b+\ii J^c{}_b\,,\\
    \lambda_{ac}C_2^{cd}\lambda_{db}&=\tfrac{1}{2}\lambda_{ac}(\Lambda^{cd}+\widetilde{\Lambda}^{cd})\lambda_{db}=\tfrac{1}{2}(\lambda_{ab}+\widetilde{\lambda}_{ab})\,,\\
    2(C^\tp_2)^{cd}\lambda_{bd}&=(\Lambda^{cd}-\widetilde{\Lambda}^{cd})\lambda_{db}=\delta^c{}_b-\ii J^c{}_b\,,\\
    \lambda_{ac}(C^\tp_2)^{cd}\lambda_{bd}
    &=\tfrac{1}{2}\lambda_{ac}(\Lambda^{cd}-\widetilde{\Lambda}^{cd})\lambda_{db}=\tfrac{1}{2}(\lambda_{ab}-\widetilde{\lambda}_{ab})\,.
\end{align}
\end{subequations}
These will be used directly in the derivation of this proof.

\textbf{Decomposition of operators.} We need to decompose the quadrature operators $\hat{\xi}^a$ into its annihilation and creation components as $\hat{\xi}^a=\hat{\xi}^a_-+\hat{\xi}^a_+$, such that with respect to $\ket{0}$ we have
\begin{align}
    \hat{\xi}^a_-\ket{0}=0\quad\text{and}\quad\bra{0}\hat{\xi}^a_+=0\,.\label{eq:decomp-ladder}
\end{align}
This decomposition can also be adopted for the following operators
\begin{subequations}
\begin{align}
    \widehat{L}_\pm &=\tfrac{1}{2}\lambda_{ac}L^c{}_b\hat{\xi}^a_\pm\hat{\xi}^b_\pm\,,
    \\
    \widehat{z}_\pm &=-z^a\lambda_{ab}\hat{\xi}^b_\pm\,.
\end{align}
\end{subequations}
They follow the same ladder properties~\ref{eq:decomp-ladder}.

\textbf{The proof.} The following derivation is valid for bosons and fermions simultaneously.

Using~\cite[Eq. (155,156)]{Hackl_2021}, we can extract the purely squeezing part 
\begin{align}\begin{split}
    e^{\widehat{K}_+}&=e^{\widehat{L}_+}e^{\frac{1}{2}\lambda_{ac}\log(\id-L^2)^c{}_b(\hat{\xi}^a_+\hat{\xi}^b_--\frac{1}{4}\Lambda^{ab})}e^{\widehat{L}_-}\\
    e^{\widehat{K}_+}\ket{0}&=e^{\widehat{L}_+}e^{\frac{1}{2}\lambda_{ac}\log(\id-L^2)^c{}_b(\hat{\xi}^a_+\hat{\xi}^b_--\frac{1}{4}\Lambda^{ab})}\ket{0}\\
    &=e^{\widehat{L}_+}e^{\pm\frac{1}{8}\lambda_{ac}\log(\id-L^2)^c{}_b\Lambda^{ba}}\ket{0}\\
    &=e^{\widehat{L}_+}\ket{0}\det(\id-L^2)^{\pm1/8}\\
    &=e^{\widehat{L}_+}\ket{0}\braket{e^{\widehat{K}_+}}\,.
\end{split}\end{align}
where note that~\cite[Eq. (174)]{Hackl_2021} tells us
\begin{align}\label{eq:<eK+>}
    \braket{J|e^{\widehat{K}_+}|J}=\begin{cases}
    \det^{\frac{1}{8}}(\id-L^2) & \textbf{(bosons)}\\
    \det^{-\frac{1}{8}}(\id-L^2) & \textbf{(fermions)}
    \end{cases},
\end{align}
\cite[Eq.~(159,160)]{Hackl_2021} provide the solution for displacements
\begin{align}\begin{split}
    \mathcal{D}(z)&=e^{\widehat{z}}=e^{\widehat{z}_++\widehat{z}_-}=e^{\widehat{z}_+}e^{\frac{1}{2}z^c\lambda_{ca}(C^\tp_2)^{ab}z^d\lambda_{db}}e^{\widehat{z}_-}\\
    \bra{0}e^{\widehat{z}}&=\bra{0}e^{\frac{1}{4}z^a(\lambda_{ab}-\widetilde{\lambda}_{ab})z^b}e^{\widehat{z}_-}
    =e^{-\frac{1}{4}z^a\widetilde{\lambda}_{ab}z^b}\bra{0}e^{\widehat{z}_-}\,.
\end{split}\end{align}

Now, the problem reduces to
\begin{align}
    \braket{0|e^{\widehat{z}}e^{\widehat{K}_+}|0}=\braket{e^{\widehat{K}_+}}e^{-\frac{1}{4}z^a\widetilde{\lambda}_{ab}z^b}\braket{0|e^{\widehat{z}_-}e^{\widehat{L}_+}|0}\,,
\end{align}
which can be solved by using~\cite[Eq. (162)]{Hackl_2021}
\begin{align}\begin{split}
    &\braket{0|e^{\widehat{z}_-}e^{\widehat{L}_+}|0}\\
    &=\bra{0}e^{(-z\lambda)_a\hat{\xi}^a_-}e^{(\frac{1}{2}\lambda L)_{bc}\hat{\xi}^b_+\hat{\xi}^c_+}\ket{0}\\
    &=\bra{0}e^{\widehat{L}_++\hat{\xi}_+}e^{(z\lambda)_aC_2^{ab}(\frac{1}{2}\lambda L)_{bc}(C_2^\tp)^{cd}(z\lambda)_d}e^{\widehat{z}_-}\ket{0}\\
    &=e^{\frac{1}{2}z^a\lambda_{ae}C_2^{ef}\lambda_{fb}L^b{}_d(C_2^\tp)^{dg}\lambda_{cg}z^c}\\
    &=e^{\frac{1}{8}z^a(\lambda_{ab}+\widetilde{\lambda}_{ab})L^b{}_d(\delta^d{}_c-\ii J^d{}_c)z^c}\\
    &=e^{\frac{1}{8}z^a(\lambda_{ab}+\widetilde{\lambda}_{ab})(L^b{}_c+\ii J^b{}_dL^d{}_c)z^c}\\
    &=e^{\frac{1}{4}z^a(g_{ab}-\ii\omega_{ab})L^b{}_cz^c}=e^{\frac{1}{4}z(g-\ii\omega)Lz}\,.
\end{split}\end{align}
A combined result for bosons is
\begin{align}\begin{split}
    &\braket{0|e^{\widehat{z}}e^{\widehat{K}_+}|0}\\&=\det(\id-L^2)^{1/8}e^{\frac{1}{4}z^ag_{ab}(L^b{}_c-\delta^b{}_c)z^c}e^{-\frac{\ii}{4}z^a\omega_{ab}L^b{}_cz^c}\,,
\end{split}\end{align}
which then complete the proof.
\end{proof}

The result~\eqref{eq:gamma} in~\eqref{eq:phase_gamma} can also be derived by applying comparisons of our operator definitions with the definitions of the same operators in~\cite[Carlton M. Caves]{cc_a_dot_pdf}. For a single bosonic degree of freedom, $\hat{\xi}^a=(\hat{q},\hat{p})$. Specifically,~\cite[Eq. (15)]{cc_a_dot_pdf} gives displacement 
\begin{align}
    D(\hat{a},\alpha)=e^{\ii(\sqrt{2}\im{\alpha}\hat{q}-\sqrt{2}\re{\alpha}\hat{p})}
\end{align}
and~\cite[Eq. (248)]{cc_a_dot_pdf} gives squeezing
\begin{align}
S(r,\phi)&=e^{\frac{\ii}{2}r[(\hat{q}\hat{p}+\hat{p}\hat{q})\cos2\phi-(\hat{q}^2-\hat{p}^2)\sin2\phi]}\,.
\end{align}
In particular, we could adopt the following identity
\begin{align}\begin{split}
    \braket{\alpha|0}_r&=\braket{0|e^{\alpha^* \hat{a}-\alpha \hat{a}^\dagger}e^{\frac{r}{2}(\hat{a}^2-\hat{a}^{\dagger2})}|0}
    \\&=\braket{0|e^{\sqrt{2}\ii(\ii\im{\alpha}\hat{q}+\re{\alpha}\hat{p})}e^{\ii\frac{r}{2}(\hat{q}\hat{p}+\hat{p}\hat{q})}|0}
    \\
    &=\frac{e^{-|\alpha|^2/2}}{\sqrt{\cosh{r}}}e^{-\frac{1}{2}(\alpha^*)^2\tanh{r}}\,,
\end{split}\end{align}
from a special form of~\cite[Eq. (275)]{cc_a_dot_pdf} when $\beta=\alpha$ and $\phi=0$. Note that here $r=\sqrt{-\det(K)}$. This also allows us to determine the function of $\gamma(M,z)$.

\section{Proof of Lemma~\ref{lemma:BCH:exp(w)exp(K)}}\label{ap:ab}

We introduce the following function in (\ref{eq:ab})
\begin{subequations}
\begin{align}
    \alpha(K)&=\frac{K}{e^{K}-\id}\,,\\
    \beta(K)&=\frac{1}{4}\frac{K-\sinh{K}}{\id-\cosh{K}}\,.
\end{align}
\end{subequations}
The function $\alpha(K)$ has the properties
\begin{align}
    \alpha^{-1}(K)=\frac{e^{K}-\id}{K}=\sum_{n=0}^\infty\frac{K^n}{(n+1)!}=\sum_{n=1}^\infty\frac{K^{n-1}}{n!}\,,\label{eq:series-alpha}
\end{align}
and 
\begin{align}
    e^K\alpha(K)=\alpha(K)e^K=\frac{K}{\id-e^{-K}}=\alpha(-K)\,,\label{eq:exp(alpha)}
\end{align}
as well as
\begin{align}
    \alpha(K)\alpha(-K)=\frac{K}{e^{K}-\id}\frac{K}{\id-e^{-K}}=\frac{1}{2}\frac{K^2}{\cosh K-\id}\,.\label{eq:alpha.alpha}
\end{align}
Finally, when the quadratic operator is zero, we have
\begin{subequations}
\begin{align}
    \lim_{K\to0}\alpha(K)&=\id\,,\\
    \lim_{K\to0}\beta(K)&=0\,,
\end{align}
\end{subequations}
satisfy the required properties. 

Now we would like to prove the Lemma~\ref{lemma:BCH:exp(w)exp(K)}.
\begin{proof}
The commutator identities tell us that the commutator between a linear operator $\widehat{w}$ and a quadratic operator $\widehat{K}$ would give us another linear operator as $[\widehat{w},\widehat{K}]=\widehat{wK}$. While the commutator between two linear operators would give us a scalar function as $[\widehat{w}_1,\widehat{w}_2]=w_2\Lambda w_1$, which then commutes with anything. The BCH-Dynkin formula (\ref{eq:Dynkin}) tells us that $\log(e^{\widehat{K}}e^{\widehat{w}})$ involves all possible commutators of all orders between $\widehat{K}$ and $\widehat{w}$. However, above arguments indicate that the order of resulting operators depends on the order of involved $\widehat{w}$. The highest-order operator is clearly just the quadratic $\widehat{K}$, which corresponds to the zeroth order of $\widehat{w}$. Thus we can write it as
\begin{align}
    \log(e^{\widehat{K}}e^{\widehat{w}})=\widehat{K}+\widehat{v}(K,w)+f(K,w)\,,
\end{align}
for some linear operator $\widehat{v}$ and scalar function $f$, correspond to the $\widehat{w}$ and $\widehat{w}^2$ terms, respectively. It follows from Dynkin's formula that there is only two types of linear terms will survive, namely
\begin{subequations}
\begin{align}
    [[\widehat{K}]^r,\widehat{w}]&=\widehat{w(-K)^r}\propto\widehat{wK^r}\,,\\
    [[\widehat{K}]^{r-1},[\widehat{w},\widehat{K}]]&=-\widehat{w(-K)^r}\propto\widehat{wK^r}\,.
\end{align}
\end{subequations}
Thus the linear terms in $\log(e^{\widehat{K}}e^{\widehat{w}})$ must have the form
\begin{align}
    \widehat{v}=\sum_{r=0}^\infty\alpha_r\widehat{wK^r}=\widehat{w\alpha(K)}\,,
\end{align}
for some power series coefficients $\alpha_r$. On the other hand, by the same Dynkin's formula, the only non-zero scalar terms must have the form
\begin{subequations}
\begin{align}
    [\widehat{w},[[\widehat{K}]^s,\widehat{w}]]&\propto[\widehat{w},\widehat{wK^s}]\,,\\
    [\widehat{w},[[\widehat{K}]^{s-1},[\widehat{w},\widehat{K}]]]&\propto[\widehat{w},\widehat{wK^s}]\,.
\end{align}
\end{subequations}
Thus the scalar function takes the form
\begin{align}\begin{split}
    f&=\sum_{s=0}^\infty\beta_s[\widehat{w},\widehat{wK^s}]=[\widehat{w},\widehat{w\beta(K)}]\\&=w\beta(K)\Lambda w=w_a\beta(K)^a{}_c\Lambda^{cb}w_b\,,
\end{split}\end{align}
for some power series coefficients $\beta_s$. 

To find $\alpha(K)$, the formalism $e^{\widehat{K}}e^{\widehat{w}}=e^{\widehat{K}+\widehat{v}+f}$ gives 
\begin{align}\begin{split}
    &(e^{\widehat{K}}e^{\widehat{w}})^{-1}\hat{\xi}^a(e^{\widehat{K}}e^{\widehat{w}})
    \\
    &=e^{-(\widehat{K}+\widehat{v}+f)}\hat{\xi}^ae^{\widehat{K}+\widehat{v}+f}
    \\
    &=e^{-(\widehat{K}+\widehat{v})}\hat{\xi}^ae^{\widehat{K}+\widehat{v}}
    \\
    &=\sum_{n=0}^\infty\frac{[\hat{\xi}^a,[\widehat{K}+\widehat{v}]^n]}{n!}
    \\
    &=(e^{K})^a{}_b\hat{\xi}^b+\sum_{n=1}^\infty\frac{(K^{n-1})^a{}_b}{n!}\frac{\Lambda^{bc}}{\iii}v_c
    \\
    &=(e^{K})^a{}_b\hat{\xi}^b+\tfrac{1}{\iii}\alpha^{-1}(K)^a{}_b\Lambda^{bc}v_c\,.
\end{split}\end{align}
Where the series (\ref{eq:series-alpha}) gives
\begin{align}
    \alpha^{-1}(K)=\frac{e^K-\id}{K}\,.
\end{align}
Recall that $\iii[\hat{\xi}^a,\widehat{v}]=\Lambda^{ab}v_b$. Hence we need to solve
\begin{align}\begin{split}
    \alpha^{-1}(K)^a{}_b\Lambda^{bc}v_c&=(e^{K})^a{}_b\Lambda^{bc}w_c\\
    \Lambda^{dc}v_c&=\alpha(K)^d{}_a(e^{K})^a{}_b\Lambda^{bc}w_c\\
    \lambda_{ad}\Lambda^{dc}v_c&=\lambda_{ad}\alpha(-K)^d{}_b\Lambda^{bc}w_c\\
    v_a&=w_c\alpha(K)^c{}_a\,,
\end{split}\end{align}
where we use the property of $\alpha$ (\ref{eq:exp(alpha)}) and the formula (\ref{eq:okof}). 

To find the explicit form of $\beta(K)$, we try to solve differential equations. Define
\begin{subequations}
\begin{align}
    A(t)&=e^{\widehat{K}}e^{t\widehat{w}}\,,
    \\
    B(t)&=e^{\widehat{K}+t\widehat{v}+t^2f}=e^{\widehat{K}+t\widehat{v}}e^{t^2f}\,.
\end{align}
\end{subequations}
Where for brevity, we save the wide-hats for $A$ and $B$. The differential equation of $A$ is clearly $A^{-1}(t)\dot{A}(t)=\widehat{w}$. Using (\ref{eq:adxy}), we can evaluate
\begin{align}\begin{split}
    &e^{-\widehat{K}-t\widehat{v}}\frac{\mathrm{d}}{\mathrm{d}t}(e^{\widehat{K}+t\widehat{v}})=\sum_{n=0}^\infty\frac{[\widehat{v},[\widehat{K}+t\widehat{v}]^n]}{(n+1)!}\\
    &=\sum_{n=0}^\infty\frac{[\widehat{v},[\widehat{K}]^n]}{(n+1)!}-tv\sum_{n=2}^\infty\frac{K^{n-1}}{(n+1)!}\Lambda v\\
    &=\sum_{n=0}^\infty\frac{\widehat{v K^n}}{(n+1)!}-tv\sum_{m=1}^\infty\frac{K^{2m-1}}{(2m+1)!}\Lambda v\\
    &=\widehat{v\alpha^{-1}(K)}-tv\frac{\sinh K-K}{K^2}\Lambda v\\
    &=\widehat{w}-tw\alpha(K)\frac{\sinh K-K}{K^2}\alpha(-K)\Lambda w\\
    &=\widehat{w}-tw\frac{1}{2}\frac{K-\sinh K}{\id-\cosh K}\Lambda w\,.
\end{split}\end{align}
In the third line, we relabel $n\to2m$ in order to remove all even terms as $v K^{2m}\Lambda v=0$ for all $m\in\mathbb{Z}$ by Lemma~2. We use (\ref{eq:alpha.alpha}) in the last line. The first sum gives a linear term while the second sum gives a scalar function. Together, to match the differential equation we must have
\begin{align}\begin{split}
    \frac{\dot{B}(t)}{B(t)}=\widehat{w}-tw\frac{1}{2}\frac{K-\sinh K}{\id-\cosh K}\Lambda w+2tw\beta(K)\Lambda w=\widehat{w}\,,
\end{split}\end{align}
this clearly shows
\begin{align}
    \beta(K)=\frac{1}{4}\frac{K-\sinh K}{\id-\cosh K}\,.
\end{align}
For the other case, we have
\begin{align}\begin{split}
    \log(e^{\widehat{w}}e^{\widehat{K}})&=\log(e^{\widehat{K}}e^{\widehat{we^K}})\\
    &=\widehat{K}+\widehat{we^K\alpha(K)}+[\widehat{we^K},\widehat{we^K\beta(K)}]\\
    &=\widehat{K}+\widehat{w\alpha(-K)}+[\widehat{w},\widehat{w\beta(K)}]\,,
\end{split}\end{align}
where we use the fact that $e^K\alpha(K)=\alpha(-K)$ again from (\ref{eq:exp(alpha)}).
\end{proof}

\section{Fermionic squeezing operators}\label{ap:fermion-sq}

In this appendix, we check the normalization condition $w_aG^{ab}w_b=2$ as well as~\ref{eq:sq_mw} defined for fermions in Section~\ref{sec:fermions}. The normalization condition is given by
\begin{align}\begin{split}
    \id &=\mathcal{S}(M_w)\mathcal{S}(M_w)
    \\
    &=w_a\hat{\xi}^aw_b\hat{\xi}^b=w_a\hat{\xi}^a\hat{\xi}^bw_b
    \\
    &=\tfrac{1}{2}(w_a\hat{\xi}^a\hat{\xi}^bw_b+w_b\hat{\xi}^b\hat{\xi}^aw_a)
    \\
    &=\tfrac{1}{2}w_a(\hat{\xi}^a\hat{\xi}^b+\hat{\xi}^b\hat{\xi}^a)w_b
    \\
    &=\tfrac{1}{2}w_a\{\hat{\xi}^a,\hat{\xi}^b\}w_b=\tfrac{1}{2}w_aG^{ab}w_b\id\,.
\end{split}\end{align}
It follow that we require $w_aG^{ab}w_b=2$. 

To check the consistency of (\ref{eq:sq_mw}) with the definition (\ref{eq:fix-S(M)-up-to-sign}), one can show
\begin{align}\begin{split}
    \mathcal{S}^\dagger(M_w)\hat{\xi}^a\mathcal{S}(M_w)&=w_c\hat{\xi}^c\hat{\xi}^aw_b\hat{\xi}^b
    \\
    &=w_c\{\hat{\xi}^c,\hat{\xi}^a\}w_b\hat{\xi}^b-\hat{\xi}^aw_c\hat{\xi}^cw_b\hat{\xi}^b
    \\
    &=w_cG^{ca}w_b\hat{\xi}^b-\hat{\xi}^a\id
    \\
    &=(w_cG^{ca}w_b-\delta^a{}_b)\hat{\xi}^b
    \\
    &=(M_w)^a{}_b\hat{\xi}^b\,.
\end{split}\end{align}
This complete the characterization of fermionic systems.

\section{Review of bosonic phase formula}\label{ap:psi}

Our Result~2 for the inhomogeneous case enables us to compute the complex phase $\Phi_J[e^{-\ii\hat{H}}]$ where $\hat{H}$ is a GQH as in~\eqref{eq:GQH}. It crucially relies on Result 3a derived in~\cite{Hackl_2024}, which considers the equivalent problem for a PQH, \ie where $-\ii \hat{H}=\widehat{K}$ for some Lie algebra element $K\in\mathfrak{sp}(2N,\mathbb{R})$.

While we refer the reader to~\cite{Hackl_2024} for the proof, we will present several self-contained versions of this result, covering specific simpler cases and the general case. In short, we ask how to compute
\begin{align}
    \Phi_J[e^{\widehat{K}}]=\frac{\braket{J|e^{\widehat{K}}|J}}{|\braket{J|e^{\widehat{K}}|J}|}\,.
\end{align}
Put differently, given a symplectic generator $K\in\mathfrak{sp}(2N,\mathbb{R})$ associated with a PQH $\hat{H}$ as $\widehat{K}=-\ii \hat{H}$, the result provides an explicit closed‐form expression for the ambiguous global phase $\arg\braket{J|e^{\widehat{K}}|J}$ for the bosonic evolution.

The difficulty of evaluating $\Phi_J[e^{\widehat{K}}]$ depends on the Jordan–Chevalley decomposition of $K$, namely there exists a unique decomposition
\begin{align}
    K=K_I+K_R+K_N\,,
\end{align}
where $K_N$ is the non-diagonalizable nilpotent part of $K$, while $K_I$ has purely imaginary eigenvalues and $K_R$ purely real eigenvalues. All three parts commute with each other. The main cause of $\Phi_J[e^{\widehat{K}}]$ \emph{wrapping around} the origin of the complex plane is $K_I$, which is why it is also useful to define
\begin{align}
    K'=K_R+K_N
\end{align}
for later use. In the following, we will now use the results of~\cite{Hackl_2024} to give \emph{simple recipes} for computing $\Phi_J[e^{\widehat{K}}]$ both numerically and analytically for the different cases.

\textbf{Numerical evaluation via evolution.} A quick numerical way to evaluate the complex phase is to use Result~1 from~\cite{Hackl_2024} stating
\begin{align}\label{eq:num-track}
    \braket{J|e^{\widehat{K}}|J}^2=\left(\overline{\det}\frac{e^{K}-Je^{K}J}{2}\right)^{-1}
\end{align}
for bosons, where $\overline{\det}$ was introduced in~\eqref{eq:def-trbar-detbar}. We would like to extract square root of this expression, which yields a sign ambiguity. However, this ambiguity can be easily lifted when we introduce an evolution parameter $t$ and compute numerically the complex phase of $\braket{J|e^{t\widehat{K}}|J}^2$ as a function of $t$. It is clear that the complex phase vanishes for $\braket{J|J}=1$ at $t=0$ and so we can track which sign is correct when taking the square root. 

\begin{table*}
  \caption{Blocks for the real Jordan normal form of the symplectic generator $K$. We denote $\lambda = \mu + \ii \nu$ (with  $\mu,\nu \in \mathbb{R}$) and $\sigma=\alpha_{\lambda}(e, \tilde{e})$, where $e$ is the generating generalized eigenvector (gGEV) associated with the block. Note that in $\mathfrak{c}=6$, $\ii\sigma \in \mathbb{R}$.}\label{TAB:NormalFormList}
	\begin{tabular}{|C{0.03\textwidth} | C{0.27\textwidth} | C{0.27\textwidth} |C{0.27\textwidth} | C{0.09\textwidth} |}
		\hline
		$\mathfrak{c}$ & $I_I^{(\mathfrak{c})} (\lambda, D)$ & $I_R^{(\mathfrak{c})}(\lambda,D)$ & $I_L^{(\mathfrak{c})}(\lambda,D)$ & dimension \\
		\hline
		1 
		&
            $\begin{pmatrix}
		\mu &  &  &   &  \\
		1 & \mu &  &  &  \\ 
		& \ddots & \ddots &  &  \\ 
		&  & & 1 & \mu \\
		\end{pmatrix}$
		&
		$\mathbb{0}$
		&
		$\mathbb{0}$
		&
		$D$
		\\
		\hline
		2
		&
		$\begin{pmatrix}
		\mu & \nu &  &  &  &  &  &  \\
		-\nu & \mu &  &  &  &  &  &  \\
		1 & 0 & \mu & \nu &  &  &  &  \\
		0 & 1 & -\nu & \mu &  &  &  &  \\
		&  &  & \ddots & \ddots &   &  &  \\
		&  &  &  & 1 & 0 & \mu & \nu \\
		&  &  &  & 0 & 1 & -\nu & \mu \\
		\end{pmatrix}$
		&
		$\mathbb{0}$
		&
		$\mathbb{0}$
		&
		$2D$
		\\
		\hline
		3
		&
		$\sigma
		\begin{pmatrix}
		0 &  &  &   \\ 
		1 & 0 &  &  \\ 
		& \ddots & \ddots & \\ 
		&   & 1 & 0\\
		\end{pmatrix} $
		&
		$\mathbb{0}$
		&
		$\sigma
		\begin{pmatrix}
		0 &  &  &   \\ 
		& 0 &  &  \\ 
		&  & \ddots & \\ 
		&   &  & (-1)^{D/2}\\
		\end{pmatrix} $
		&
		$D/2$ (integer)
		\\
		\hline
		4
		&
		$\begin{pmatrix}
		0 &  &  &   \\ 
		1 & 0 &  &  \\ 
		& \ddots & \ddots & \\ 
		&   & 1 & 0\\
		\end{pmatrix} $
		&
		$\mathbb{0}$
		&
		$\mathbb{0}$
		&
		$D$ (odd)
		\\
		\hline
		5
		&
		$\mathbb{0}$
		&
		$\sigma
		\begin{pmatrix}
		&  & &  &  & \nu\\
		&  &  & & \nu & 1\\
		&  &  & & -1 & \\
		&  &\iddots  & \iddots &  & \\
		& \nu & -1 &  & & \\
		\nu & 1 &  & &  & \\
		\end{pmatrix}$
		& 
		$\sigma
		\begin{pmatrix}
		&  & &  &  -1 & -\nu\\
		&  &  & 1& -\nu & \\
		&  &\iddots  & \iddots &  & \\
		& 1 &  & & & \\
		-1 & -\nu &  &  & & \\
		-\nu & &  & &  & \\
		\end{pmatrix}$
		&
		$D$ (even)
		\\
		\hline
		6
		&
		$\begin{pmatrix}
		0 &  &  &   \\ 
		1 & 0 &  &  \\
		& \ddots & \ddots & \\ 
		&   & 1 & 0\\
		\end{pmatrix} $
		&
		$ \ii \sigma
		\begin{pmatrix}
		&  &  &  & \nu\\
		&  &  & -\nu & \\
		&  & \iddots & & \\
		& -\nu &  &  & \\
		\nu &  &  &  & \\
		\end{pmatrix} $
		&
		$\ii \sigma
		\begin{pmatrix}
		&  &  &  & -\nu\\
		&  &  & \nu & \\
		&  & \iddots &  & \\
		& \nu &  &  & \\
		-\nu &  &  &  & \\
		\end{pmatrix}$
		&
		$D$ (odd)
		\\
		\hline
	\end{tabular}
\end{table*}

\textbf{Closed form for $K=K_I$ if $K_I\Omega>0$.} Given a diagonalizable generator $K=K_I$ with purely imaginary eigenvalues and the property that $K_I\Omega$ is positive-definite, we can follow these steps:
\begin{enumerate}
    \item Find a valid complex structure $\tilde{J}$, such $\tilde{J}\Omega \tilde{J}^\intercal=\Omega$, $\Omega \tilde{J}>0$ and $[K_I,\tilde{J}]=0$. For this, we can choose
    \begin{align}\label{eq:JtildeKI}
        \tilde{J}=K_I |K_I|^{-1}\,,
    \end{align}
    where $|K_I|$ refers to replacing all eigenvalues of $K_I$ by their absolute value. Note that $|K_I|$ and $K_I$ commute. It is easy to see that $\tilde{J}^2=-\id$, as multiplying with $|K_I|^{-1}$ changes all eigenvalues of $K_I$ to be $\pm \ii$.
    \item Compute the transformation $T=\sqrt{-\tilde{J}J}$, where $-\tilde{J}J$ has purely positive eigenvalues, so that the square root is well-defined.
    \item Evaluate the complex phase as
    \begin{align}
        \arg\Phi_J[e^{\widehat{K}}]=\tfrac{1}{4}\Tr(K_I\tilde{J})+\tfrac{1}{2}\eta(T^{-1}e^{K},T)\,.
    \end{align}
\end{enumerate}
\emph{If $K_I\Omega$ is not positive-definite, the expression for $\tilde{J}$ in~\eqref{eq:JtildeKI} will not satisfy $\tilde{J}\Omega>0$ and thus not be a valid complex structure. However, this can be easily fixed by flipping the signs of some of the eigenvalues of $K_I$.}

\textbf{General case.} Given a general generator $K$, we can follow the steps explained  to write
\begin{enumerate}
    \item Follow the steps explained in~\cite{Kustura:2019PhRvA..99b2130K} to find a change of basis $M$, such that $MKM^{-1}$ takes the form
    \begin{align}\label{KN}
    \begin{pmatrix} O_I & O_R \\ O_L & -O_I^\intercal  \end{pmatrix}\,,
    \end{align}
    where $O_\chi \in \mathbb{R}^{N\times N} $ ($\chi=I,L,R$), $O_R= O_R^\intercal$, and $O_L= O_L^\intercal$. The matrices $O_\chi$ can each be expressed as a direct sum of blocks over different eigenvalue types
    \begin{align}\label{eq:MKMform}
       O_\chi \equiv  O_\chi^{(\mathfrak{\mathcal{R}})} \oplus O_\chi^{(\mathfrak{\mathcal{C}})} \oplus O_\chi^{(0)} \oplus O_\chi^{(\mathfrak{\mathcal{I}})}\,,
    \end{align}
    where
    \begin{align}
    \begin{split}
    O_\chi^{(\mathfrak{\mathcal{R}})} &\equiv \bigoplus_{\mathcal{R}_i   }  \spare{\Moplus_{j=1}^{m_i} I_{\chi}^{(1)}(\lambda_i,D_{ij})},\\
    O_\chi^{(\mathfrak{\mathcal{C}})} &\equiv \bigoplus_{\mathcal{C}_i }  \spare{\Moplus_{j=1}^{m_i} I_{\chi}^{(2)}(\lambda_i,D_{ij})},\\
    O_\chi^{(0)} &\equiv  \bigoplus_{j=1}^{l_0} I_{\chi}^{(3)}(0,D_{0j}) \bigoplus_{j=1}^{n_0} I_{\chi}^{(4)}(0,D_{0j}),\\
    O_\chi^{(\mathcal{I})} &\equiv \bigoplus_{\mathcal{I}_i  }  \spare{ \Moplus_{j=1}^{l_i} I_{\chi}^{(5)}(\lambda_i,D_{ij}) \Moplus_{j=1}^{n_i} I_{\chi}^{(6)}(\lambda_i,D_{ij})}\hspace{-1cm}
    \end{split}
    \end{align}
    with eigenvalues $\lambda_i$ and associated Jordan block dimensions $D_{ij}$ for the Jordan block $j$. The matrices $I^{(\mathfrak{c})}_\chi(\lambda,D)$ can be read from table~\ref{TAB:NormalFormList} where $\mathfrak{c}=1,\dots,6$ labels six distinct cases.
    \item Define the complex structure $\tilde{J}$, such that $M\tilde{J}M^{-1}$ takes the matrix form
    \begin{align}
        \begin{pmatrix}
            0 & \id\\
            -\id & 0
        \end{pmatrix}
    \end{align}
    in the same basis, where $MKM^{-1}$ takes the form~\eqref{eq:MKMform}.
    \item Compute the transformation $T=\sqrt{-\tilde{J}J}$, where $-\tilde{J}J$ has purely positive eigenvalues, so that the square root is well-defined.
    \item Evaluate the complex phase as
    \begin{align}
\begin{split}\label{eq:main-bosons}
    \arg\braket{J|e^{\widehat{K}}|J}&=\tfrac{1}{4}\Tr(K_I\tilde{J})+\tfrac{1}{2}\eta(T^{-1},e^{K})\hspace{-7mm}\\
    &\quad -\arg\overline{\det}\sqrt{T^{-1}\tfrac{e^{K'}-\tilde{J}e^{K'}\tilde{J}}{2}T}\hspace{-7mm}\\
    &\quad +\tfrac{1}{2}\eta(T^{-1}e^{K},T)\,.
\end{split}
\end{align}
\end{enumerate}

In summary, Result 3a from~\cite{Hackl_2024} using the standard form from~\cite{Kustura:2019PhRvA..99b2130K} gives a robust and general expression for the phase of the Gaussian vacuum overlap $\arg\braket{J|e^{\widehat{K}}|J}$ in the bosonic case. Its careful decomposition and cocycle accounting make the complex phase accessible for analytical calculations. The relevant expressions simplify for many generators $K$ encountered in practice, such as those representing physical Hamiltonians. If one is instead mostly interested in numerically evaluating the phase, one can track it efficiently using the formula~\eqref{eq:num-track}.

\bibliography{references.bib}
\end{document}